\documentclass[11pt]{article}
\usepackage{amsmath}
\usepackage{amssymb}
\usepackage{amscd}
\usepackage{amsthm}
\usepackage{amsfonts}
\usepackage{thmtools}
\usepackage{thm-restate}
\usepackage{graphicx}
\usepackage[usenames,dvipsnames]{xcolor}
\usepackage{hyperref}
\usepackage[noabbrev,capitalize]{cleveref}
\usepackage[
    backend=biber,
    style=alphabetic,
    maxnames=50
]{biblatex}

\usepackage{geometry}
\usepackage{multirow}

\usepackage{algorithm}
\usepackage{algpseudocode}

\usepackage{caption}
\usepackage{subcaption}
\usepackage{wrapfig}

\algrenewcommand{\algorithmicrequire}{\textbf{Input:}}
\algrenewcommand{\algorithmicensure}{\textbf{Output:}}

\hypersetup{
    colorlinks,
linkcolor={orange},
    citecolor={OliveGreen},
urlcolor={blue!80!black}
} 

\newcommand{\rr}{\mathbb{R}}
\newcommand{\OO}{O}
\newcommand{\OOt}{\widetilde{\OO}}
\newcommand{\abs}[1]{\left\lvert #1 \right\rvert}

\newcommand{\psnorm}[2][p]{\left\lVert #2 \right\rVert_{#1}}

\global\long\def\AA{\boldsymbol{\mathit{A}}}

\global\long\def\CC{\boldsymbol{\mathit{C}}}

\global\long\def\PP{\boldsymbol{\mathit{P}}}

\global\long\def\vb{\boldsymbol{b}}
\global\long\def\vc{\boldsymbol{c}}
\global\long\def\vd{\boldsymbol{d}}

\global\long\def\vp{\boldsymbol{p}}
\global\long\def\vq{\boldsymbol{q}}

\global\long\def\vv{\boldsymbol{v}}
\global\long\def\vw{\boldsymbol{w}}
\global\long\def\vx{\boldsymbol{x}}
\global\long\def\vy{\boldsymbol{y}}

\global\long\def\vdelta{\boldsymbol{\delta}}

\global\long\def\ones{\boldsymbol{1}}
\global\long\def\zeros{\boldsymbol{0}}
\global\long\def\Tapprox{\mathcal{T}_{\textnormal{approx}}}
\global\long\def\Texact{\mathcal{T}_{\textnormal{exact}}}
\global\long\def\Tmixed{\mathcal{T}_{\textnormal{mixed}}}
\global\long\def\Tlinear{\mathcal{T}_{\textnormal{linear}}}

\global\long\def\smax{\textnormal{smax}}
\global\long\def\opt{\textnormal{OPT}}
\global\long\def\nnz{\textnormal{nnz}}

\global\long\def\Vmax{V_{\max}}

\newtheorem{theorem}{Theorem}
\newtheorem*{theorem*}{Theorem}
\newtheorem{lemma}[theorem]{Lemma}
\newtheorem*{lemma*}{Lemma}

\newtheorem{corollary}[theorem]{Corollary}
\newtheorem{definition}{Definition}
\newtheorem{problem}{Problem}

\newtheorem*{remark*}{Remark}
\newtheorem{claim}{Claim}

\crefname{claim}{Claim}{Claims}
\Crefname{claim}{Claim}{Claims}
 
\usepackage{mathrsfs}
\usepackage[scr=boondox,  cal=esstix]   {mathalpha}

\usepackage{xparse}

\usepackage[english]{babel}

\definecolor{darksepia}{rgb}{0.44, 0.26, 0.08}
\definecolor{sepia}{rgb}{0.91,0.86,0.8}
\renewcommand{\epsilon}{\varepsilon}

\floatstyle{ruled}
\newfloat{procedure}{tbp}{lop}
\floatname{procedure}{Procedure}

\crefname{procedure}{Procedure}{Procedures}
\Crefname{procedure}{Procedure}{Procedures}

\renewcommand*{\theHtheorem}{\theHsection.\the\value{theorem}}

\title{Adaptive Sparsification for Linear Programming}
\author{
    \'{E}tienne Objois\footnote{IRIF, Université Paris Cité, \texttt{objois@irif.fr}} 
    \and
    Adrian Vladu\footnote{CNRS, IRIF, Université Paris Cité, \texttt{vladu@irif.fr}}
}
\date{}

\DeclareSourcemap{
  \maps[datatype=bibtex]{
    \map{
      \step[fieldset=urldate, null]
    }
  }
}

\begin{document}
\maketitle

\begin{abstract}
We introduce a generic framework for solving linear programs (LPs) with many constraints $(n \gg d)$ via adaptive sparsification. Our approach provides a principled generalization of the techniques of [Assadi '23] from matching problems to general LPs and robustifies [Clarkson's '95] celebrated algorithm for the exact setting. The framework reduces LP solving to a sequence of calls to a ``low-violation oracle'' on small, adaptively sampled subproblems, which we analyze through the lens of the multiplicative weight update method.

Our main results demonstrate the versatility of this paradigm. First, we present a quantum version of Clarkson's algorithm that finds an exact solution to an LP using $\tilde{O}(\sqrt{n} d^3)$ row-queries to the constraint matrix. This is achieved by accelerating the classical bottleneck (the search for violated constraints) with a generalization of Grover search, decoupling the quantum component from the classical solver. Second, our framework yields new state-of-the-art algorithms for mixed packing and covering problems when the packing constraints are ``simple''. By retaining all packing constraints while sampling only from the covering constraints, we achieve a significant width reduction, leading to faster solvers in both the classical and quantum query models. Our work provides a modular and powerful approach for accelerating LP solvers.

\end{abstract}

\section{Introduction}

Linear Programming (LP) is a central primitive in optimization, with applications ranging from combinatorial optimization to quantum information. Classical research has long pursued faster solvers, both in the high-precision regime (polynomial in $\log 1/\varepsilon$) via interior point and cutting plane methods~\cite{vaidyaSpeedingupLinearProgramming1989,leeEfficientInverseMaintenance2015,cohenSolvingLinearPrograms2019,brandSolvingTallDense2021}, and in the low-precision regime (polynomial in $1/\varepsilon$) using first-order techniques such as multiplicative weights~\cite{plotkinFastApproximationAlgorithms1991,grigoriadisSublineartimeRandomizedApproximation1995,aroraMultiplicativeWeightsUpdate2012}.

A particularly intriguing regime arises when the number of constraints $n$ far exceeds the dimension $d$. Since only $d$ constraints suffice to pin down an optimal solution, one may hope to reduce $n$ dramatically without losing accuracy. In the context of graph algorithms, this idea underlies sparsification techniques~\cite{benczurRandomizedApproximationSchemes2015,spielmanGraphSparsificationEffective2008}, which replace large inputs with much smaller sketches that preserve the structure of (near-)optimal solutions. However, its potential for linear programming remains largely unexplored.

A classical result of Clarkson~\cite{clarksonVegasAlgorithmsLinear1995} gives a concrete way to reduce constraints adaptively. Instead of attempting a one-shot reduction of all $n$ constraints, Clarkson iteratively samples small subsets, solves the resulting reduced LPs, and uses the solutions to adaptively adjust sampling probabilities. This iterative reweighting guarantees that after only $\widetilde{O}(d)$\footnote{$\OOt(.)$ hides poly-logarithmic dependencies on $n$, $d$, and $\epsilon^{-1}$.} iterations, the algorithm identifies the true solution while never solving an LP with more than $O(d^2)$ constraints. More recently, Assadi~\cite{assadiSimple$1varepsilon$ApproximationSemiStreaming2025} showed how related ideas yield powerful semi-streaming algorithms for approximate matching, with performance reminiscent of multiplicative weights but with strictly better dependence on the error parameter $\varepsilon$. Together, these results suggest that adaptive sparsification could form a general principle for reducing large LP instances to efficiently solvable ones.

This perspective is particularly timely in the quantum setting. In the row-query model, solving an LP to constant precision requires at least $\Omega(\sqrt{nd})$ queries to the constraint matrix~\cite{apeldoornQuantumSDPSolversBetter2020}, a bound conjectured to be tight~\cite{apersQuantumSpeedupsLinear2024}. Moreover, quantum algorithms can sample $r$ elements from a distribution over $n$ elements in just $\widetilde{O}(\sqrt{nr})$ queries~\cite{apersQuantumSpeedupGraph2023}. Clarkson’s and Assadi's frameworks, which only needs $\OOt(d)$ and $\OOt(\frac{1}{\epsilon})$ iterations of sampling a small number of constraints each, therefore naturally aligns with quantum speedups.

These observations lead to a guiding question:
\begin{center}
\emph{Can adaptive sparsification provide a unifying framework for quantum algorithms that efficiently solve LPs exactly or approximately?}
\end{center}

\subsection{Our Contributions} 

A linear program is defined as 
\begin{equation}
    \label{eq:intro-lp}
    \max_{\vx \in \mathcal{D}} \langle \vc, \vx \rangle, \qquad \text{subject to } \AA \vx \leq \vb,
\end{equation}
where $\AA \in \rr^{n \times d}$, $\vb \in \rr^n$ and $\vc \in \rr^d$. $\mathcal{D}$ should be seen as a simple convex set, present to upper bound the \emph{width} of the linear program. We denote by $\opt$ the optimum value, that is $\max \lbrace \langle \vc, \vx \rangle : \AA \vx \leq \vb \text{ and } \vx \in \mathcal{D}\rbrace$.

We first present a quantum version of Clarkson's algorithm that uses a generalization of Grover Search to sample the constraints achieves low query complexity of $\OOt(\sqrt{n}d^3)$. We present the comparison with other algorithms in \cref{tab:exact-solvers}.

\begin{theorem}
    \label{thm:clarkson-quantum-intro}
    There is a randomized quantum algorithm that finds an exact solution of \eqref{eq:intro-lp} using $\OOt(\sqrt{n} d^3)$ row-queries to $\AA$. Furthermore, this algorithm runs in expected time
    \begin{equation*}
        \OOt\left(\sqrt{n} d^3 r + d \cdot \Texact(d,24d^2)\right),
    \end{equation*}
    where $r$ is the row-sparsity of $\AA$, and $\Texact(d,n)$ is the time required to compute an exact solution to a linear program with $d$ variables and $n$ constraints.
\end{theorem}

\begin{remark*}
    When the bit complexity of the input is not large, one can rely on high-precision solvers such as interior point methods. If $L$ is the bit-complexity of the input, that is, the number of bits required to write $(\AA,\vb,\vc)$,\footnote{Assume $\AA \in \rr^{n \times d}$, $\vb \in \rr^n$ and $\vc \in \rr^d$ are integral, and $n \gg d$. One can think of $L$ being the same order of magnitude as $\log n + \log (1 + d_{\max}) + \log (1 + \max \lbrace \psnorm[\infty]{\vc}, \psnorm[\infty]{\vd}\rbrace)$ where $d_{\max}$ is the largest subdeterminant of $\AA$ in absolute value.} then we can combine Clarkson's algorithm with~\cite{leeEfficientInverseMaintenance2015, vandenbrandMinimumCostFlows2021} to achieve an expected time of $\OOt(d \cdot \nnz(\AA) + L\cdot d^{3.5} \min \lbrace r, \sqrt{d} \rbrace)$. In this case, our algorithm still requires $\OOt(\sqrt{n}d^3)$ row-queries to $\AA$, and runs in time $\OOt(\sqrt{n} d^3 r + L\cdot d^{3.5} \min \lbrace r, \sqrt{d} \rbrace)$.
\end{remark*}

\begin{table}
    \centering
    \begin{tabular}{c c c c }
        \hline
        \textbf{Paper} & \textbf{Query model} & \textbf{Total runtime} & \textbf{Query complexity} \\ \hline
        \multicolumn{4}{c}{exact solver} \\ \hline
        \cite{clarksonVegasAlgorithmsLinear1995} & classical & $\OO\left(d \log n \cdot (\nnz(\AA) + \Texact(d,6d^2))\right)$ & - \\ \hline
        \cref{thm:clarkson-quantum} & quantum & $\OOt\left(d (\sqrt{n}d^2 r + \Texact(d,24d^2))\right)$ & $\OOt\left(\sqrt{n}d^3\right)$ \\ \hline
        \multicolumn{4}{c}{high-precision solver} \\ \hline
        \cite{vandenbrandMinimumCostFlows2021} & classical & $\OOt\left(nd + d^{2.5} \right)$ & - \\ \hline
        \cite{leeEfficientInverseMaintenance2015} & classical & $\OOt\left(\nnz(\AA) \sqrt{d} + d^{2.5}\right)$ & - \\ \hline
        \cite{apersQuantumSpeedupsLinear2024} & quantum & $\OOt\left(\sqrt{n} d^{7} r^2 + d^{\omega + 7}\right)$ & $\OOt\left(\sqrt{n}d^5 \right)$ \\ \hline
        quantum version  & 
        \multirow{2}{*}{quantum} &  \multirow{2}{*}{$\OOt\left(\sqrt{n}d r+ d^3 \right)$} &  \multirow{2}{*}{$\OOt\left(\sqrt{n}d \right)$}  \\ 
        of~\cite{leeFasterCuttingPlane2015} \\ \hline
\end{tabular}
    \caption{Comparison between algorithms for solving a linear program of the form $\max \langle \vc, \vx \rangle$ subject to $\AA \vx \leq \vb$ and $\vx \in \mathcal{D}$. We consider $\AA \in \rr^{n \times d}$, $r$ is the row-sparsity of $\AA$, $\rho = \max_{\vx \in \mathcal{D}} \psnorm[\infty]{\AA \vx - \vb}$ is the width of the LP, $\epsilon >0$, and $\omega$ is the matrix multiplication exponent.
    }
    \label{tab:exact-solvers}
\end{table}

\paragraph{Low-Precision.} In the low-precision regime, we give an algorithm that computes an $\epsilon$-approximate solution using $\OO(\Vmax/\epsilon \log n)$ iterations, where $\Vmax$ corresponds to the largest violation achievable by a solution in $\mathcal{D}$. $\Vmax$ can be much smaller than the width, for pure covering problems for instance, $\Vmax = 1$, while the width can be as large as $d$.

We provide two quantum versions of this classical algorithm. The main difference comes from computing the query access to the sampling probabilities. In the first version, we require fewer iterations, but querying the sampling probability of a constraint at iteration $t$ requires $t$ row-queries to $\AA$. In the second version, we perform more iterations, but querying the sampling probability of a constraint requires only one row-query to $\AA$. We present the different results in \Cref{tab:lp-solvers}.

\begin{table}
    \centering
    \begin{tabular}{c c c c}
        \hline
        \textbf{Paper} & \textbf{Query model} & \textbf{Type} & \textbf{Total runtime} \\ \hline
\cite{aroraMultiplicativeWeightsUpdate2012} & classical & deterministic & $\OO\left(\frac{\rho \Vmax}{\epsilon^2} \log n \cdot (\nnz(\AA) + \Tlinear(d))\right)$ \\ \hline
\cref{thm:lp-classical} & classical & randomized & $\OOt\left(\frac{\Vmax}{\epsilon} (\nnz(\AA) + \Texact(d,24d \frac{\Vmax}{\epsilon}))\right)$ \\ \hline
        \cref{thm:lp-quantum} & quantum & randomized & $\OOt\left(\frac{\Vmax}{\epsilon} (\sqrt{n d \frac{\Vmax}{\epsilon}} \frac{\Vmax}{\epsilon} r + \Texact(d,24d \frac{\Vmax}{\epsilon}))\right)$ \\ \hline
        \cref{thm:lp-quantum-width} & quantum & randomized & $\OOt\left(\frac{\rho}{\epsilon} (\sqrt{n d \frac{\rho}{\epsilon}} r + \Texact(d,24d \frac{\rho}{\epsilon}))\right)$ \\ \hline
    \end{tabular}
    \caption{Comparison between algorithms for finding an $\epsilon$-approximate solution to a linear program of the form $\max \langle \vc, \vx \rangle$ subject to $\AA \vx \leq \vb$ and $\vx \in \mathcal{D}$. We consider $\AA \in \rr^{n \times d}$, $r$ is the row-sparsity of $\AA$, $\rho = \max_{\vx \in \mathcal{D}} \psnorm[\infty]{\AA \vx - \vb}$ is the width of the LP, $\epsilon >0$, and $\Vmax = \max_{\vx \in \mathcal{D}} \max_{i \in [n]} \left(\AA \vx - \vb\right)_i$ is the largest violation possible.
    }
    \label{tab:lp-solvers}
\end{table}

\paragraph{Mixed Packing and Covering Problems.} A mixed packing and covering problem is of the form
\begin{equation*}
    \text{find } \vx \in [0,1]^d, \qquad \text{subject to } \CC \vx \geq \ones, \PP \vx \leq \ones,
\end{equation*}
where $\CC \in [0,1]^{n_c \times d}$, and $\PP \in [0,1]^{n_p \times d}$.\footnote{Solving a mixed packing and covering LP can be reduced to multiple problems of this form by adding a poly-logarithmic number of constraints, see appendix B in~\cite{boobFasterWidthdependentAlgorithm2019}.} When the matrix $\CC$ is tall, i.e. $n_c \gg n_p + d + \epsilon^{-1}$, we give a randomized algorithm improving the current state-of-the-art. We also provide a quantum version of this algorithm. We present those results in \cref{tab:mpc-solvers}.

\begin{table}
    \centering
    \begin{tabular}{c c c c}
        \hline
        \textbf{Paper} & \textbf{Query model} & \textbf{Type} & \textbf{Total runtime} \\ \hline
\cite{youngNearlyLinearWorkAlgorithms2014,quanrudNearlyLinearTime2019} & classical & deterministic & $\OOt\left(\frac{\nnz}{\epsilon^2}\right)$ \\ \hline
        \cite{boobFasterWidthdependentAlgorithm2019} & classical & deterministic & $\OOt\left(\frac{r}{\epsilon} \nnz\right)$ \\ \hline
        \cite{chekuriRandomizedMWUPositive2018} & classical & randomized & $\OOt\left(\frac{\nnz}{\epsilon} + \frac{n}{\epsilon^2} + \frac{d}{\epsilon^3}\right)$ \\ \hline
        \cref{thm:mpc-classical} & classical & randomized &$\OOt\left(\frac{1}{\epsilon} \cdot (\nnz(\CC) + \Tmixed(d,n_p, \OOt\left(\frac{d}{\epsilon}\right), \epsilon))\right)$ \\ \hline
        \cref{thm:mpc-quantum} & quantum & randomized & $\OOt \left(\frac{1}{\epsilon} \left(\sqrt{n_c \frac{d}{\epsilon}} \frac{r_c}{\epsilon} + \Tmixed(d,n_p, \OOt\left(\frac{d}{\epsilon}\right), \epsilon)\right)\right)$ \\ \hline
    \end{tabular}
        \caption{Comparison between algorithms finding a $1+\epsilon$ approximate solution of mixed packing and covering problems of the form: find $\vx \in [0,1]^d$ subject to $\PP \vx \leq \ones$ and $\CC \vx \geq \ones$. We assume $\PP \in [0,1]^{n_p \times d}$ and $\CC \in [0,1]^{n_c \times d}$. Here $r_p$ (resp. $r_c$) denotes the row-sparsity of $\PP$ (resp. $\CC$). $n = n_p + n_c$, $r = \max \lbrace r_p, r_c \rbrace$, and $\nnz = \nnz(\PP) + \nnz(\CC)$. $\Tmixed(d,n_p,n_c,\epsilon)$ denotes the running time of a $1+\epsilon$ approximate mixed packing and covering LP solver on $d$ variables with $n_p$ (resp. $n_c$) packing (resp. covering) constraints.
    }
    \label{tab:mpc-solvers}
\end{table}

\begin{remark*}
 Assume there is a constant number of packing constraints (but strictly larger than $1$, otherwise it becomes a pure covering LP with other, faster solvers). In this case, we may have $\nnz(\PP) = \OO(d)$, assume $\CC$ is tall and dense $\nnz(\CC) = n \cdot d$, and $n \gg d^2\epsilon^{-2} + \epsilon^{-3}$. In that case, if we use the algorithm from~\cite{quanrudNearlyLinearTime2019} at each iteration, we find an $\epsilon$-approximate solution in time $\OOt\left(\frac{nd}{\epsilon} + \frac{d^2}{\epsilon^4}\right)$ which is faster than other mixed packing and coverings problems. The only draw back is our algorithm is randomized while \cite{youngNearlyLinearWorkAlgorithms2014,quanrudNearlyLinearTime2019, boobFasterWidthdependentAlgorithm2019} are deterministic.
\end{remark*}

 \subsection{Overview of Approach}
\label{ssec:overview}

Our core contribution is a unified and versatile framework for solving linear programs (LPs) through \emph{constraint sparsification}. We present a single, iterative algorithm that uses the multiplicative weights update (MWU) method as its central engine. By simply adjusting its internal parameters, this one algorithm can function as either a low-precision approximate solver or an exact solver. Our framework consists of $4$ steps performed iteratively.

\begin{enumerate}
    \item \emph{Sparsify:} Construct a sparse LP by selecting a small, weighted subset of constraints.
    \item \emph{Solve:} Find a solution for the sparse LP.
    \item \emph{Verify:} Check the solution against the entire original set of $n$ constraints to find violations.
    \item \emph{Refine:} Update the weighting scheme to increase the importance of the violated constraint and return to step $1$, creating a new, refined sparsifier.
\end{enumerate}

In step $1$, we sparsify the linear program by sampling constraints from a distribution. In step $2$ we solve the partial linear program obtained, in this step we show that we can either use an exact solver, or an $\epsilon$-approximation one. We then need to verify which constraint has been violated, this is used to compute a new distribution (step $4$) for the next iteration. This framework is also exceptionally amenable to quantum speedups. One difficulty we face is that the distribution changes a lot from one iteration to the other so maintaining it can be challenging. Fortunately, we are making few iterations which allows us to recompute it from scratch each iteration.

\paragraph{Constraint sparsification for low-precision algorithms.}\label{sec:overview_constr_sparf} We first present our approach for approximately solving LPs via an adaptive sequence of calls to an exact LP solver on a \emph{sparse} subset of constraints. The algorithm generalizes Assadi's approach~\cite{assadiSimple$1varepsilon$ApproximationSemiStreaming2025} for maximum matching to general settings, and our analysis casts his approach in a principled framework.

Formally, suppose we aim to approximate the decision problem $\{\AA \vx \leq \vb, \vx \in \mathcal{D}\}$, in the sense that we seek to either return some $\vx \in \mathcal{D}$ such that $\AA \vx \leq \vb + \epsilon \ones$, or return $\textsc{infeasible}$ if the input problem does not have a solution.

The textbook approach for solving this problem is via the MWU framework put forth in~\cite{plotkinFastApproximationAlgorithms1991}, which runs in $\OO\left(\rho^2 \log n/\epsilon^2 \right)$ iterations, each of which involves calling a linear optimization oracle over $\mathcal{D}$, and updating the weights fed into the oracle in time $\OO(\nnz(\AA))$. Here $\rho = \max_{\vx \in \mathcal{D}} \Vert \AA\vx -\vb \Vert_\infty$ represents the \emph{width} of the problem.

In our case, we instead solve a sequence of $\OO(\rho \log n / \epsilon)$ sub-problems, each of which involves solving exactly an LP over a sparse subset of constraints. The approach is as follows: in each iteration we sample a subset of constraints $Q$, and exactly solve the LP containing only the constraints in $Q$. Then we check the solution, and for each of the original constraints that is violated, we increase its sampling probability. To understand why this approach works, we can cast it in the MWU framework by defining a \emph{low-violation oracle}, namely a procedure which given a distribution $\vp \in \Delta_n$ of constraint weights, outputs a solution $\vx \in \mathcal{D}$ such that the \emph{violation vector} $\vv = \ones_{\AA\vx > \vb}$ represented by the indicator vector of constraints violated by $\vx$, violates only constraints with low total weight, in the sense that $\langle \vp, \vv \rangle \leq \frac{\epsilon}{3\rho}$. Simply running MWU on these low-weight violation vectors for $T = \OO\left(\rho \log n/\epsilon \right)$ iterations will yield a sequence $\{\vv_t\}_{t=1}^T$
for which $\max \left\{ \frac{1}{T}\sum_{t=1}^T \vv_t \right\} \leq \epsilon/\rho$. This in turn translates into a bound on $\max \left\{\AA \left(\frac{1}{T} \sum_{t=1}^T \vx_t\right) -\vb\right\} \leq \epsilon$, since each constraint violation corresponds to an iterate $\vx_t$ for which the corresponding violated constraint $\AA_{i:}$, satisfies $\langle \AA_{i:}, \vx_t \rangle - \vb_i \leq \rho$. So if each constraint is violated at most $\left(\epsilon/\rho\right) \cdot T$ times, over all the $T$ iterations, the total violation suffered by that constraint is at most $\epsilon T$.

The advantageous feature of this approach is that, since the loss vectors $\vv_t$ that we run MWU have only $0$-$1$ entries, and we tolerate a large (constant) step size, the number of iterations required scales only linearly with $\rho/\epsilon$ rather than quadratically as in standard instantiations of~\cite{plotkinFastApproximationAlgorithms1991}. We discuss this further in \cref{sec:mwubasics}.

We can obtain a low-violation oracle simply by sub-sampling the constraints, and exactly solving an LP for the sampled subset. Indeed, per the analysis in~\cite{clarksonVegasAlgorithmsLinear1995}, to achieve this we can simply sample each constraint $i$ with probability proportional to $\vp_i$. The main challenge is bounding the number of sampled constraints required, which we show is the product of the number of variables and the number of iterations.

We can use an $\epsilon$-approximation solver as long as the solver's output set contains few elements. Assadi used this technique on some combinatorial problems, known to have few solutions, such as matching~\cite{assadiSimple$1varepsilon$ApproximationSemiStreaming2025}. When the solver's output set is sparse, we can still sample constraint $i$ with probability proportional to $\vp_i$. Per the analysis in~\cite{assadiSimple$1varepsilon$ApproximationSemiStreaming2025}, we can upper bound the number of sampled constraints by the logarithm of the size of the solver's output set. In particular, for LPs for which a sparse $\epsilon$-net covers all the possible outputs of the approximate LP solver subroutine exists, we can sample a number of constraints proportional to the logarithm of this $\epsilon$-net's size.

We observe that for \emph{nice} solution space, a sparse $\epsilon$-net can be obtained. For mixed packing and covering problems, we can assume a solution has entry of the form $\vx_i = \epsilon/d (1+\epsilon)^{k_i}$ for some integer $k_i \in [1 + \frac{1}{\epsilon} \log \frac{d}{\epsilon}]$. This allows us to obtain an $\epsilon$-net of size $\OO((\epsilon^{-1} \log d/\epsilon)^d)$. We prove we can reduce a mixed packing and covering LP by solving $\OOt(\frac{1}{\epsilon})$ times a sparser version of the LP containing all the packing constraints, but only $\OO(\frac{d}{\epsilon} \cdot \log (\frac{1}{\epsilon} \log \frac{d}{\epsilon}))$ of its covering constraints.

\paragraph{A unified framework for constraint sparsification.}

Consider the decision problem $\lbrace \AA \vx \geq \vb, \vx \in \mathcal{D} \rbrace$. We generalize the analysis for low-precision algorithms to exact solvers. We in fact prove that it is sufficient to find a sequence of solutions $\lbrace \vx_j \rbrace_{j \in [T]}$ such that 
\begin{equation}
    \label{eq:sequence}
    \abs{\lbrace j \in [T] : \langle \AA_{i:}, \vx_j \rangle > \vb_i +\epsilon \rbrace} < \mu \cdot T,
\end{equation}
for some value of $\epsilon$ and $\mu$. This abstraction is the foundation of our work. Per Clarkson's analysis, solving this problem with $\epsilon = 0$, and $\mu = d$ guarantees one of the $\vx_j$ is an exact solution. On the other hand, per Assadi's analysis, it is sufficient to find solutions for $\epsilon$, and $\mu = \epsilon$ to obtain $(1+2 \epsilon)$-approximation of pure covering problems.

Solving \ref{eq:sequence} fits naturally inside the MWU framework. Consider $f_i : X \mapsto \abs{\lbrace \vx \in X : \langle \AA_{i:}, \vx \rangle > \vb_i + \epsilon \rbrace}$. The goal is to find a set $X$ of solutions such that $\max_{i \in [n]} f_i(X) < \mu \abs{X}$.

Our main algorithm finds the desired set of solutions $\{\vx_j\}_{j \in [T]}$ iteratively. In each of the $T$ iterations, we use an \emph{MWU-based procedure as a black box} to find the next solution, $\vx_j$. This procedure maintains a distribution of weights over the $n$ constraints. Its goal is to call a base LP solver on a sparse sample of the constraints to find a candidate solution $\vx_j$ that satisfies the ``low-violation'' property relative to the current weights. That is, the solution primarily satisfies the constraints that the MWU framework currently deems important (i.e., those with high weight). MWU reduces finding $\lbrace \vx_j \rbrace_{j \in [T]}$ satisfying \ref{eq:sequence} to iteratively finding $\vx$ such that  for a given $\vp \in \Delta_n$, $\langle \vp, \ones_{\AA \vx > \vb + \epsilon \ones} \rangle < \mu$.

In order to solve this easier problem, we rely on sampling constraints and using a base solver on the sampled constraints. We consider what we call $(\epsilon,\mu)$-low-violation oracles. Given $\vp \in \Delta_n$, a low-violation oracle returns $\vx \in \rr^d$ such that $\langle \vp, \ones_{\AA \vx > \vb + \epsilon \ones} \rangle < \mu$. Hence, we obtain the following theorem.

\begin{theorem}[Short version of \ref{thm:main-theorem}]
    It is possible to find a set of vectors $\lbrace \vx_j \rbrace_{j \in [T]}$ solution to \ref{eq:sequence} using $T = \OO(\frac{\log n}{\mu})$ calls to an $(\epsilon,\mu)$-low-violation oracle.
\end{theorem}

\paragraph{Exploiting quantum speedups via Grover search.}

We show that our framework is extremely amenable to quantum speedups. In fact, since the running time of our algorithms is largely dominated by the time required to identify violated constraints, being able to efficiently do so could lead to important improvements.

We show in \cref{sec:quantum} that in the quantum query access model, where each constraint has an associated feasibility oracle given a candidate solution $\vx$, one can replace the time dependence in input sparsity with $\sqrt{n} d^{\OO(1)}$, which yields sublinear time algorithms when $n \gg d$. While previous works~\cite{apersQuantumSpeedupGraph2023,apersQuantumSpeedupsLinear2024} have developed algorithms that are carefully crafted to fit the restrictive requirements of quantum computational models, in our case the source of speedups stems uniquely from quickly identifying violated constraints. This is achieved via a generalization of Grover's algorithm, which is the only quantum power tool we rely on (c.f. \cref{lem:apers}).

Our main challenge is to access the sampling probabilities. Our sampling probabilities are proportional to the weights inside the MWU framework. For a set of solutions $X$, the weight of a constraint is an exponential of the number of solutions from $X$ violating this same constraint. Since we update the weights \emph{aggressively}, any constant upper bound is no longer true after a constant number of iterations. This prevents us from reducing the update time of the sampling probabilities even in the classical scheme. With quantum query access, we do not maintain the weight value from one iteration to the other, hence we need to recompute it at each iteration. That means at iteration $t$, we require $t$ row-query access to the constraint matrix. In the $\epsilon$-approximation setting, we show that we can actually use one query access to $\AA$ at each iteration, but this comes at the cost of a larger width.

We also need to estimate the sum of the weights to a constant factor at each iteration, this is performed at each iteration by sampling few constraints.

\subsection{Organization}

The remainder of this paper is organized as follows. We begin in \cref{sec:preliminaries} by establishing the necessary preliminaries, including our notation, formal definitions for the classes of linear programs we consider, and the quantum query model used in our analysis. In \cref{sec:paradigm}, we introduce our unified iterative sparsification paradigm, which forms the theoretical core of our work. We formalize the problem of finding a set of solutions with a low frequency of constraint violations and introduce the key algorithmic primitive, the low-violation oracle, which underpins our entire framework. \cref{sec:applications} demonstrates the versatility of this framework in the classical setting; we show how to construct these oracles and use them to generalize the seminal results of Clarkson and Assadi, to efficient algorithms for low-precision general LPs and mixed packing and covering problems. Finally, in \cref{sec:quantum}, we show that our framework is amenable to quantum acceleration, demonstrating how replacing the classical bottleneck of identifying violated constraints with a quantum search subroutine leads to significant speedups and improved query complexities.
 \section{Preliminaries}
\label{sec:preliminaries}

\subsection{Notation}

We write vectors and matrices in bold. Matrices are written in uppercase. We use $\langle . , . \rangle$ notation to denote the inner product between two vectors. For a matrix $\AA \in \rr^{n \times d}$ and an integer $i \in [n]$, $\AA_{i:}$ denotes the $i$\textsuperscript{th} row of $\AA$. For a subset $S \subseteq [n]$, we let $\AA_{S}$ be the matrix $\AA$ whose rows are restricted to those in set $S$. We define $\nnz(\AA)$ as the number of nonzero entries of $\AA$. For a linear program of the form \ref{eq:general-lp}, there is a set $B$ of $d$ constraints such that the solution $\vx$ satisfies $\AA_B \vx = \vb_B$, which we call the base of the linear program.

\subsection{Classes of Linear Programs}

In the literature, LPs have been analyzed under different assumptions. Consider a linear program of the form $\lbrace \vx : \AA \vx \leq \vb \rbrace$. The first important remark on solving LPs is whether we have an \emph{implicit} or \emph{explicit} access to $\AA$. The classic MWU approach can be used with implicit access as long as a linear optimization oracle exists. Often, implicit LPs arise when the number variables $d$ is exponential while the number of constraint is not. This can be the case for spanning tree covering problems for instance, the number of spanning trees on a graph is exponential while there is only one constraint per edge. In this work, we focus on LPs given \emph{explicitly}, meaning $\AA$ can be stored in the (Q)RAM.

\paragraph{General LPs.} In general, a linear program can be written as:
\begin{equation}
    \label{eq:general-lp}
    \max \langle \vc, \vx \rangle, \qquad \text{subject to } \AA \vx \leq \vb.
\end{equation}
We denote by $\opt$ the optimal value of \eqref{eq:general-lp}. For exact solvers, we consider LPs that are not degenerate, that means there is only one set of $d$ constraints $B \subseteq [n]$ such that the optimum of the linear program lies at their intersection. An $\epsilon$-feasible solution of \ref{eq:general-lp} is a vector $\vx$ such that $\AA \vx \leq \vb + \epsilon \ones$. An $\epsilon$-approximation of \ref{eq:general-lp} is an $\epsilon$-feasible solution $\vx$ such that $\langle \vc, \vx \rangle \geq \opt$.

\paragraph{Positive LPs.}

Positive LPs are of the form \ref{eq:general-lp} with an added constraint on the sign of rows of $\AA$. We distinguish three types of positive LPs, first \emph{pure covering} LPs are of the form $\min \lbrace \psnorm[1]{\vx} : \CC \vx \geq \ones \rbrace$ where $\CC$ has positive entries. Second, there are \emph{pure packing} LPs with the form $\max \lbrace \psnorm[1]{\vx} : \PP \vx \leq \ones \rbrace$ where $\PP$ has positive entries. Note that the dual of a pure covering LP is a pure packing LP. Finally, there are mixed packing/covering LPs, they are of the form $\lbrace \vx : \CC \vx \geq \ones, \PP \vx \leq \ones \rbrace$.

By adding a poly-logarithmic number of variables, we can assume that $\CC$ and $\PP$ have entries in $[0,1]$. Moreover, we can assume that the exact solution is $\vx$ such that $\zeros \leq \vx \leq \ones$ coordinate-wise.

We say that $\vx$ is a $(1+\epsilon)$-approximation of a covering problem if $(1+\epsilon) \opt \geq \psnorm[1]{\vx}$, and $\CC \vx \geq \ones$. Note that if we have $\tilde{\vx}$ such that $\psnorm[1]{\tilde{\vx}} \leq \opt$ and $\CC \tilde{\vx} \geq (1-\epsilon)\ones$, then $\vx = \tilde{\vx}/(1-\epsilon)$ is a $(1+2\epsilon)$-approximation.

\subsection{Quantum Complexity Model}

To analyze the performance of our quantum algorithms, we adopt a standard quantum query model, consistent with recent literature on quantum algorithms for optimization and linear algebra~\cite{apersQuantumSpeedupGraph2023,apersQuantumSpeedupsLinear2024}. In this model, the input data of the linear program (specifically the constraint matrix $\AA \in \mathbb{R}^{n \times d}$ and the vector $\vb \in \mathbb{R}^n$) is not assumed to be stored in active memory (RAM). Instead, we access the constraints through queries to a quantum oracle.

The fundamental operation we consider is the verification of a single constraint against a candidate solution vector. For any given classical vector $\vx \in \mathbb{R}^d$, we assume access to a \emph{feasibility oracle}, denoted $O_{\vx,\epsilon}$. This oracle acts on the indices of the $n$ constraints and identifies whether a specific constraint is violated by $\vx$. Formally, the oracle is a unitary transformation that acts on the standard basis as follows:
\begin{equation*}
    O_{\vx,\epsilon} : |i\rangle|z\rangle \mapsto |i\rangle|z \oplus \vv_i^\epsilon(\vx)\rangle \quad \forall i \in [n], z \in \{0,1\}
\end{equation*}
where $\oplus$ denotes addition modulo $2$, and $\vv^\epsilon_i(\vx)$ is a binary indicator of violation:
\begin{equation*}
    \vv_i^\epsilon(\vx) = 
    \begin{cases} 1 & \text{if } \langle \AA_{i:}, \vx \rangle > \vb_i + \epsilon \quad \text{(constraint } i \text{ is violated)} \\ 0 & \text{if } \langle \AA_{i:}, \vx \rangle \leq \vb_i + \epsilon \quad \text{(constraint } i \text{ is satisfied)} \end{cases}
\end{equation*}

The primary metric for the complexity of our quantum algorithms is the \emph{quantum query complexity}, defined as the number of calls made to oracles of this type.

The implementation of the oracle $O_{\vx,\epsilon}$ for a given $\vx$ requires the ability to compute the inner product $\langle \AA_{i:}, \vx \rangle$ and compare it to $\vb_i + \epsilon$. This, in turn, requires access to the data of the $i$-th row of $\AA$ and the $i$-th entry of $\vb$. We therefore model each query to $O_{\vx, \epsilon}$ as a single \emph{row query} to the pair $(\AA, \vb)$, which provides the necessary information to evaluate the $i$-th constraint.

We will use quantum query access to $(\AA,\vb)$ to sample the constraints faster than in the classical query model. The execution of the classical LP solver on small, sampled sub-problems, is performed on a classical computer. This hybrid model allows us to precisely isolate the source of quantum speedup, which, in our framework, stems from the ability to find a violated constraint (an index $i$ for which $\vv_i^\epsilon(\vx) = 1$) more efficiently than is possible classically.

 \section{The Iterative Sparsification Framework}
\label{sec:paradigm}

In this section, we present the general sparsification framework. Formally, our goal is to solve the following problem. In \cref{sec:applications} we present how to reduce LP problems to \cref{prob:intro1}.

\begin{restatable}{problem}{intro}
    \label{prob:intro1}
    Let $\epsilon \geq 0$, $\mu > 0$, and suppose we have a linear program $(\AA,\vb,\vc)$ with optimal value $\opt$. In the \emph{low-violation solution set problem}, we must find a set $\lbrace \vx_j \rbrace_{j = 1}^T$ such that:
    \begin{enumerate}
        \item $\langle \vc, \vx_j \rangle \geq \opt$ for all $j \in [T]$;
        \item For every constraint $i \in [n]$, the set of solutions violating it is small: 
        \[
        \abs{\lbrace j \in [T] : \langle \AA_{i:}, \vx_j \rangle > \vb_i + \epsilon \rbrace} < \mu \cdot T\,.\]
\end{enumerate}
\end{restatable}

We show that it is possible to solve \cref{prob:intro1} by iteratively using a low-violation oracle (\cref{def:lvo}). The central result of this section is the following.

\begin{theorem}
    \label{thm:main-theorem}
    Consider a linear program of the form
    \begin{equation*}
        \max_{\vx} \langle \vc, \vx \rangle, \qquad \text{subject to } \AA \vx \leq \vb.
    \end{equation*}
    Assume we are given a $(\epsilon,\mu/3)$-low-violation oracle with running time $\mathcal{T}$, then in $T = \OO\left(\frac{\log n}{\mu}\right)$ iterations, we are able to compute a solution to \cref{prob:intro1} $\lbrace \vx_1, \ldots, \vx_T \rbrace$. Moreover, each iteration runs in time $\OO(\nnz(\AA) + \mathcal{T})$.
\end{theorem}

Based on this, we can now focus exclusively on efficiently implementing a low-violation oracle. To do so it suffices to solve an LP on a sub-sampled set of constraints. To ensure that this procedure yields a low-violation oracle we need to sample sufficiently many constraints (the number depends on the output space), and we need to sample them with appropriate probabilities. It turns out that these probabilities are exactly those used by the MWU routine. For the quantum implementation of our algorithms, we will leverage the following structural property of these weights.

\begin{claim}
    \label{cl:prob_propto_exp}
    In iteration $t$, we have $\vp_t \propto 2^{\sum_{\tau \in [t]} \vv^\epsilon (\vx_\tau)}$, where $\vv_i^\epsilon(\vx)$ as $1$ if $\langle \AA_{i:}, \vx \rangle > \vb_i + \epsilon$, and $0$ otherwise.
\end{claim}

\subsection{The Abstract Framework}
\label{ssec:reduction}

We would like to sample from a set of $n$ constraints, a set of $r$ constraints with $r \ll n$ such that a solution on these $r$ constraints gives us a solution of the LP. Such a one-shot sparsification technique is unlikely to exists as mentioned in the introduction. A simpler problem is to find solutions $\vx_1, \ldots, \vx_T$ such that any constraint is violated by a small portion of $\lbrace \vx_j \rbrace_{j \in [T]}$. Formally, instead of sparsifying the LP directly, we want to solve the following problem. 

\intro*

First, we remark that \cref{prob:intro1} is easier than finding an exact solution of the LP or even an $\epsilon$-approximation. Indeed, if $\vx$ is an exact solution of the linear program, then $\lbrace \vx \rbrace$ is a solution to \cref{prob:intro1}. Also, if $\vx$ is an $\epsilon$-approximation, then $\lbrace \vx \rbrace$ is a solution to \cref{prob:intro1}.

\cref{prob:intro1} is easier since we accept $T \geq 2$. With $\mu < 1$, we should have $T \geq \mu^{-1}$. We will prove we can construct a set of solutions with $T = \OO(\frac{\log n}{\mu})$. While it could be possible, it is unlikely for one to be able to find the $T$ solutions without them being correlated. Meaning, we select the sets of constraints without computing a solution in between. Assume one seeks an exact solution to the LP using an exact solver. In order to create the $T$ solutions, one may use the exact solver on sampled constraints. Using \cref{lem:sampling2} we can upper bound the average number of violated constraint, however this does not mean anything on the probability of a specific constraint to be violated. In fact, constraints from the basis are far more likely to be violated. Instead of trying to find the set of $T$ solutions in a one-shot fashion, we build them iteratively using the MWU framework.

\begin{theorem}
    \label{thm:mwu}
    Let $f : \Delta_n \mapsto \lbrace 0,1 \rbrace^n$ be a (randomized) routine which given any probability distribution $\vp \in \Delta_n$, it outputs a vector $\vv \in \lbrace 0,1 \rbrace^n$ such that $\langle \vp, \vv \rangle \leq \mu$. Then one can provide a sequence $\lbrace \vp_t \rbrace_{t=1}^T$ such that the corresponding outputs $\lbrace \vv_{t} \rbrace_{t=1}^T$ satisfy 
    \begin{equation*}
        \max \left\lbrace \frac{1}{T} \sum_{t=1}^{T} \vv_t \right\rbrace \leq 3 \mu,
    \end{equation*}
    where $T = \OO(\frac{\log n}{\mu})$.
\end{theorem}

We can solve \cref{prob:intro1} using \cref{thm:mwu}, given any probability distribution $\vp \in \Delta_n$, we need to find a solution $\vx$ such that $\langle \vp, \vv^\epsilon(\vx) \rangle \leq \mu$ where $\vv^\epsilon_i(\vx)$ is $1$ if $\langle \AA_{i:}, \vx \rangle > \vb_i +\epsilon$ and $0$ otherwise.

\subsection{Low-Violation Oracles}

Low-violation oracles are used inside our MWU framework. Given any probability distribution, we seek to find $\vx \in \rr^d$ such that
\begin{equation}
    \label{eq:mwu}
    \begin{cases}
        \langle \vc, \vx \rangle &\geq \opt; \\
        \langle \vp, \vv^{\epsilon}(\vx) \rangle &\leq \mu.
    \end{cases}
\end{equation}
Where $\vv_i^{\epsilon}(\vx)$ is $1$ if $\langle \AA_{i:}, \vx \rangle > \vb_i +\epsilon$ and $0$ otherwise.

\begin{definition}
    \label{def:lvo}
    A \emph{$(\epsilon, \mu)$-low-violation oracle} is an oracle such that given any probability distribution $\vp \in \Delta_n$ outputs a solution $\vx$ such that $\langle \vp,\vv^\epsilon(\vx) \rangle \leq \mu$.
\end{definition}

We can prove \cref{thm:main-theorem} using iteratively a low-violation oracle.

\begin{proof}[Proof of \cref{thm:main-theorem}]
   The proof will rely on the MWU method, more precisely on \cref{thm:mwu}. 
      \cref{thm:mwu} provides the sequence $\lbrace \vp_t \rbrace_{t=1}^T$ such that, if at iteration $t$ we are able to compute $\vv_t$ with $\langle \vp_t, \vv_t \rangle \leq \mu$, then $\sum_{t=1 }^T \vv_t \leq 3 T \mu \ones$.
  Given any $\vp_t$, we use a $(\epsilon,\mu)$-low-violation oracle to compute $\vx_t$ such that $\langle \vp_t, \vv^{\epsilon}(\vx_t) \rangle \leq \mu$. Hence, after $T$ iterations, we have $\max \left\lbrace \sum_{t=1}^{T} \vv^{\epsilon}(\vx_t) \right\rbrace \leq 3 T \mu$.
   Thus, with $T = \OO(\log n/\mu)$ calls to a $(\epsilon,\mu/3)$-low-violation oracle, it is possible to solve \cref{prob:intro1}.
\end{proof}

 \section{Applications}
\label{sec:applications}

In this section, we focus on corollaries of \cref{thm:main-theorem}. We first develop the techniques from~\cite{clarksonVegasAlgorithmsLinear1995,assadiSimple$1varepsilon$ApproximationSemiStreaming2025} to design $(\epsilon,\mu)$-low-violation oracles. We then show that we can extend our analysis to find low-precision solutions for general LPs, and how to perform \emph{width} reduction for mixed packing and covering problems. We present the proofs in \cref{apps:applications}.

\subsection{Constructing \texorpdfstring{$(\epsilon,\mu)$}{(epsilon,mu)}-Low-Violation Oracles}

In this subsection, we construct two $(\epsilon,\mu)$-low-violation oracles. Both of them rely on an inner LP solver, which is used on a sparse subset of constraints. We present low-violation oracles that can be decomposed of three steps:
\begin{enumerate}
    \item \emph{Sparsify:} Construct a sparse LP by selecting a small subset of constraints using the probability vector $\vp$.
    \item \emph{Solve:} Use the solver to find a solution to the sparse LP.
    \item \emph{Verify:} Verify the solution violates constraints with small total probability. 
\end{enumerate}

When the solver used in the subroutine is an exact solver, we use Clarkson's sampling lemma (\cref{lem:sampling2} from \cite{clarksonVegasAlgorithmsLinear1995}) to bound the quantity of sampled constraint required. When we use any other solver that has a sparse output set, we bound the number of sampled constraint using a technique similar to Assadi~ \cite{assadiSimple$1varepsilon$ApproximationSemiStreaming2025}.

Our first $(\epsilon,\mu)$-low-violation oracle arises from using an exact solver on a small set of constraints. This result is a corollary of Clarkson's sampling lemma.\footnote{See \cref{lem:sampling2} for the definition of the sampling lemma.}

\begin{lemma}[Corollary of~\cite{clarksonVegasAlgorithmsLinear1995}]
    \label{lem:lvo-clarkson}
    Consider a linear program of the form $\max \langle \vc, \vx \rangle$ subject to $\AA \vx \leq \vb$, where $\AA \in \rr^{n \times d}$. For a parameter $\mu$, given a probability distribution $\vp \in \Delta_n$, the following procedure is a $(\epsilon,\mu)$-low-violation oracle.
    \begin{enumerate}
        \item \emph{Sparsify:} Consider $S$ a subset of constraints such that every $i$ is in $S$ independently with probability at least $\min \lbrace 2 d/\mu \cdot \vp_i, 1 \rbrace$.
        \item \emph{Solve:} Let $\vx$ be the exact solution of the LP subject to constraint set $S$.
\item \emph{Verify:} If $\langle \vp, \vv^\epsilon(\vx) \rangle > \mu$ go back to $1.$, else return $\vx$.
    \end{enumerate}
    Moreover, this procedure takes expected time $\OO(\nnz(\AA) + \Texact(d,2 d/ \mu))$, where $\Texact(d,n)$ is the time required to find an exact solution of a linear program with $d$ variables and $n$ constraints.
\end{lemma}

Similarly, we use the same technique as used in~\cite{assadiSimple$1varepsilon$ApproximationSemiStreaming2025} to prove that as long as the solver has a small output space, the number of sampled constraints does not have to be large.

\begin{lemma}
    \label{lem:lvo-assadi}
    Consider a linear program of the form
    \begin{equation*}
        \max \langle \vc, \vx \rangle, \qquad \text{subject to } \AA \vx \leq \vb.
    \end{equation*}
    For $\epsilon \geq 0$ and $\mu > 0$, assume we have an $\epsilon$-approximate solver $\mathcal{A}$ such that the number of distinct outputs of $\mathcal{A}$ is upper bounded by $N$. The following procedure describes a $(\epsilon,\mu)$-low-violation solver and require one call to $\mathcal{A}$ with high-probability.
    \begin{enumerate}
        \item \emph{Sparsify:} Consider $S$ a subset of constraints such that every $i$ is in $S$ independently with probability at least $\min \lbrace \frac{\log (Nn)}{\mu} \cdot \vp_i, 1 \rbrace$.
        \item \emph{Solve:} Compute $\vx$ using $\mathcal{A}$ on the LP made of constraints from $S$.
        \item \emph{Verify:} If $\langle \vp, \vv^\epsilon(\vx) \rangle > \mu$ go back to $1.$, else return $\vx$.
    \end{enumerate}
    Moreover, this procedure takes expected time $\OO(\nnz(\AA) + \Tapprox(d,\frac{\log (Nn)}{\mu}, \epsilon))$, where $\Tapprox(d,n,\epsilon)$ is the time required to find an $\epsilon$-approximate solution of a linear program with $d$ variables and $n$ constraints.
\end{lemma}

\subsection{Solving LPs to Low-Precision}

In this subsection, we present how to obtain an $\epsilon$-approximation using multiple calls to an exact solver on a sparse subset of constraints.

Consider a linear program $\max_{\vx \in \mathcal{D}} \langle \vc, \vx \rangle$ subject to $\AA \vx \leq \vb$. Let $\opt$ denote the optimal value for that LP. We would like to obtain an $\epsilon$-approximation, that means $\vx$ such that $\langle \vc, \vx \rangle \geq \opt$ and $\AA \vx \leq \vb + \epsilon\ones$. First, we reduce this problem to a problem of the form \cref{prob:intro1}. For a solver $\mathcal{A}$ whose outputs are in $\mathcal{D}$, we consider the largest violation possible in $\mathcal{D}$: $\Vmax := \max_{\vx \in \mathcal{D}} \max_{i \in [n]} (\AA \vx - \vb)_i$.

\begin{remark*}
    We can compare $\Vmax$ to the \emph{width} parameter in the Plotkin-Shmoys-Tardos (PST) framework, and other classical techniques using MWU~\cite{plotkinFastApproximationAlgorithms1991,aroraMultiplicativeWeightsUpdate2012}. In those techniques, the width $\rho$ is defined as $\max_{\vx \in \mathcal{D}} \psnorm[\infty]{\AA \vx - \vb}$. Hence, we always have $\Vmax \leq \rho$. $\rho$ not only captures the largest violation, but also the largest slack, that means, if a solution can ``largely satisfy'' a constraint, the width is large. We develop this in \cref{ssec:mixed-packing-covering}. \end{remark*}

\begin{claim}
    \label{cl:lp-classical}
    For $\epsilon > 0$, and a linear program of the form \eqref{eq:general-lp}. Let $\Vmax$ be the largest violation of $\mathcal{D}$, $\Vmax := \max_{\vx \in \mathcal{D}} \max_{i \in [n]} \left(\AA \vx - \vb\right)_i$. Given an $(\epsilon,\epsilon/3\Vmax)$-low-violation oracle $\mathcal{A}$. Within $\OO(\frac{\Vmax}{\epsilon} \log n)$ iterations, each one requiring one call to $\mathcal{A}$ and $\nnz(\AA)$ time, it is possible to find $\vx$ such that $\AA \vx \leq \vb + 2\epsilon \ones$.
\end{claim}

\begin{proof}
    Using \cref{thm:main-theorem}, it is possible to use only $T = \OO(\frac{\Vmax}{\epsilon} \log n)$ calls to an $(\epsilon,\epsilon/ 3\Vmax )$-low-violation oracle to solve \cref{prob:intro1} with parameters $\epsilon$ and $\mu = \epsilon/\Vmax$.

    Assume we have a solution $\lbrace \vx_j \rbrace_{j \in [T]}$ to \cref{prob:intro1} with those parameters. For a constraint $i$, let $V_i \subseteq [T]$ be the indices of solutions violating constraint $i$: $\lbrace j \in [T] : \langle \AA_{i:}, \vx_j \rangle \geq \vb_i + \epsilon \rbrace$. We have
    \begin{equation*}
        \sum_{j \in [T]} \langle \AA_{i:}, \vx_j \rangle - \vb_i = \sum_{j \in V_i}  \langle \AA_{i:}, \vx_j \rangle - \vb_i + \sum_{j \in [T] \setminus V_i}  \langle \AA_{i:}, \vx_j \rangle - \vb_i \leq \sum_{j \in V_i} \Vmax + \sum_{j \in [T] \setminus V_i} \epsilon \leq 2 \epsilon \cdot T.
    \end{equation*}
    Hence, $\bar{\vx} := \frac{1}{T} \sum_{j \in [T]} \vx_j$ is $2\epsilon$-feasible. Moreover, we have $\langle \vc, \bar{\vx} \rangle \geq \opt$ since for any $j \in [T]$, $\langle \vc, \vx_j \rangle \geq \opt$, hence $\bar{\vx}$ is a $2\epsilon$-approximation.
\end{proof}

Using an exact solver and \cref{lem:lvo-clarkson}, we create an $(\epsilon/2,\epsilon/6\Vmax)$-low-violation oracle that we use with \cref{cl:lp-classical}. This gives us the following theorem.

\begin{theorem}
    \label{thm:lp-classical}
    For a linear program of the form \eqref{eq:general-lp} with largest violation \[\Vmax := \max_{\vx \in \mathcal{D}} \max_{i \in [n]} \left(\AA \vx - \vb\right)_i\,,\] there exists a randomized algorithm that runs in expected time
    \begin{equation*}
        \OO\left(\frac{\Vmax}{\epsilon} \log n \left(\nnz(\AA) + \Texact(d,6d \frac{\Vmax}{\epsilon})\right)\right),
    \end{equation*}
    where $\Texact(d,n)$ is the running time of an exact solver on $d$ variables and $n$ constraints.
\end{theorem} 

\subsection{Solving Mixed Packing and Covering LPs to Low-Precision}
\label{ssec:mixed-packing-covering}

For a linear program, the largest violation $\Vmax := \max_{\vx \in \mathcal{D}} \max_{i \in [n]} \left(\AA \vx - \vb\right)_i$ corresponds to a single-sided width. Generally, the width is $\rho := \max_{\vx \in \mathcal{D}} \psnorm[\infty]{\AA \vx - \vb}$. There are linear programs for which $\Vmax \ll \rho$, for instance, pure covering problems. In a pure covering LP, the width $\rho$ can be as large as $d$, while our single-sided width is constant, equal to $1$. On the other hand, for pure packing LPs, the standard definition of the width and our single-sided width are equal.

We show how we can leverage this different notion of width to achieve better guarantees for mixed packing and covering problems with numerous dense covering constraints.

\begin{problem}
    \label{prob:mixed-lp}
    For positive integers $n_c, n_p$, and $d$. Consider $\CC \in [0,1]^{n_c \times d}$, and $\PP \in [0,1]^{n_p \times d}$.\footnote{This assumption can be made without loss of generality by adding poly-logarithmic number of constraints, see Appendix B of~\cite{boobFasterWidthdependentAlgorithm2019}.} We denote by $r_c$ the row-sparsity of $\CC$, and $r_p$ the row-sparsity of $\PP$.
    
    In the \emph{mixed packing and covering program} problem, we seek $\vx \in [0,1]^d$ such that $\PP \vx \leq \ones$, and $\CC \vx \geq \ones$, or certify that there is none. 
\end{problem}

A $(1+\epsilon)$-approximation of \cref{prob:mixed-lp} is a vector $\vx \in [0,1]^{d}$ such that $\CC \vx \geq \ones$, and $\PP \vx \leq (1+\epsilon) \ones$. Given a problem of the form \ref{prob:mixed-lp}, the one-sided width is at most $r_p$.

\paragraph{Reducing the one-sided width.} We can reduce the one-sided width by keeping some constraints. In the classical PST framework, a lot of iterations are performed, each one requires to solve a linear optimization problem within $\mathcal{D}$. Since a lot of iterations are required, the structure of $\mathcal{D}$ should be ``simple''. This comes at the cost of a potentially large width. In our case, $\mathcal{D}$ is the output set of the solver used in each call to the low-violation oracle. Since we are performing few iterations, it can be beneficial to add some constraints that are kept between iterations. In particular, we can add the constraints with large one-sided width to it. For our mixed packing and covering problem, we add the packing constraints, we consider:
\begin{equation*}
    \mathcal{D} := \lbrace \vx \in [0,1]^d : \PP \vx \leq \ones \rbrace.
\end{equation*}

Since the packing constraints are captured by the solution space, we can simulate solving a linear program with constant one-sided width, and solution space $\mathcal{D}$.

Consider a solver $\mathcal{A}$ such that given any subset $S \subseteq [n_c]$ of constraints, $\mathcal{A}$ finds a $1+\epsilon$ approximation $\vx$ of the mixed packing and covering program defined by $\CC_S$ and $\PP$. We have for any packing constraint: $\langle \PP_{i:}, \vx \rangle - 1 \leq \epsilon$, and for any covering constraint: $1 - \langle \CC_{i:}, \vx \rangle \leq 1$. Hence, the one-sided width of such a solver is $1$. 

\paragraph{Reducing the cardinality of the output space.} The next difficulty we have to handle is using \cref{lem:lvo-assadi}. \cref{lem:lvo-assadi} requires the output space to contain few elements. Assadi leverages the small cardinality of the vertex cover problem in a bipartite graph. For mixed packing and covering problems, the cardinality of the solution space can be large, however we can ``discretize'' the output space. Assume $\vx$ is a $(1+\epsilon)$-approximation, then any vector $\vy$ such that $\vx \leq \vy \leq (1+\epsilon) \vx$ coordinate-wise is a $(1+4 \epsilon)$-approximation. Hence, we can reduce the output set to vectors of the form $\vx_i = C (1+\epsilon)^{k_i}$ for some integer $k_i$, and a small value $C$. We give the size of this $\epsilon$-net in the following claim.

\begin{claim}
    \label{cl:sparsify-mixed}
    For $0 < \epsilon \leq 1$. Given $\vx$, a $1 +\epsilon$ approximation of a mixed packing and covering problem of the form \cref{prob:mixed-lp}, it is possible to find $\tilde{\vx}$ a $1 + 4 \epsilon$ approximation of the same mixed packing and covering program such that $\tilde{\vx} \in \mathcal{S}$ where $\abs{\mathcal{S}} \leq {\lceil 2 + 4\frac{\log (r_p/\epsilon)}{\epsilon} \rceil}^d$.
\end{claim}

Note that it is not necessary to have a solver that outputs in $\mathcal{S}$. This is an artifact of the analysis. With our two ingredients, we can obtain the following theorem.

\begin{theorem}
    \label{thm:mpc-classical}
    Given a mixed packing and covering program \ref{prob:mixed-lp}, it is possible to find with high-probability a $1+\OO(\epsilon)$-approximation $\vx$ or certify infeasibility in time 
    \begin{equation*}
        \OOt \left(\frac{1}{\epsilon} \left(\nnz(\CC) + \Tmixed(d,\nnz(\PP), \OO\left(\frac{d r_c}{\epsilon} \log( 2+ 4\frac{\log \frac{r_p}{\epsilon}}{\epsilon} )\right), \epsilon)\right)\right),
    \end{equation*}
    where $\Tmixed(d, \nnz(\PP), \nnz(\CC), \epsilon)$ is the time required to compute a $1+\epsilon$-approximation of a mixed packing and covering problem with at most $\nnz(\PP)$ (resp. $\nnz(\CC)$) non-zeros inside the packing (resp. covering) constraint matrix, and $d$ variables.
\end{theorem}

We can use state-of-the-art solvers for mixed packing and covering programs such as \cite{chekuriRandomizedMWUPositive2018, boobFasterWidthdependentAlgorithm2019, quanrudNearlyLinearTime2019}.

 \section{Speedups via Quantum Sampling}
\label{sec:quantum}

In the classical regime, the running time is dominated by the time needed to find violated constraints. Given a quantum query access oracle and Grover search we can improve the running time of the algorithm. Given a word $\vx \in \lbrace 0, 1 \rbrace^n$, the hamming weight of $\vx$ is its number of non-zero entries. If we have query access to a word $\vx$ with known hamming weight $t$, Grover search finds an $i$ such that $\vx_i = 1$ in time $\OO(\sqrt{n/t})$. However, if the hamming weight of $\vx$ is unknown, then Grover search becomes probabilistic. 
Instead of using Grover search to find the violated constraints, we will use it to sample from the probability vector $\vp$ at each iteration. One bottleneck of our approach is that at each iteration, we increase the cost of querying from $\vp$. At iteration $t$, we have $t$ solutions to check violation for, hence we need to perform $t$ row-queries to $(\AA,\vb)$. We strongly rely on \cref{lem:apers} to sample the constraints using queries to the binary violation oracle.

\begin{lemma}[Claim 3, \cite{apersQuantumSpeedupGraph2023}]
    \label{lem:apers}
    Assume we have query access to a list of probabilities $(\vq_i)_{i \in [n]}$, each of which is described with $\OOt(1)$ bits of precision. Then there is a quantum algorithm that samples a subset $S \subseteq [n]$, such that $S$ contains every $i$ independently with probability $\vq_i$, in expected time $\OOt(\sqrt{n \sum_i \vq_i})$ and using a QRAM of $\OOt(\sqrt{n \sum_i \vq_i})$ bits.
\end{lemma}

We give the proofs in \cref{apps:quantum_claims}.

\subsection{Sampling the Constraints}

Consider now a quantum low-violation oracle. Our low-violation oracles sample constraints using quantum queries to $(\AA,\vb)$, and then use a classical solver on the sampled constraints. 

A low-violation oracle receives a vector $\vp$ on which it has quantum query access. Assume for some value $s > 1$, the low-violation oracle requires to sample every constraint $i$ independently with probability $s \vp_i$. Using \cref{lem:apers}, we can perform this sampling step with $\OOt(\sqrt{ns})$ queries to $\vp$. In the classical regime, we had access to $\vp$. In the quantum regime, we only have query access to the binary violation vector $\vv^{\epsilon}$.

Instead of sampling from $s \vp$, we will sample from $s \vq$ where $\vq$ satisfies $\vp \leq \vq \leq 2 \vp$. This ensures that on average, $2s$ constraints are sampled, hence the number of queries to $\vp$ and to $\vq$ in \cref{lem:apers} remains the same up to a constant factor.

We now describe how we obtain our query access to $\vq$. Consider $\vw$ the vector proportional to $\vp$ but not normalized. By \cref{cl:prob_propto_exp}, at iteration $t$, we have $\vw_i = 2^{\sum_{\tau = 1}^{t} \vv_i^{\epsilon}(\vx_\tau)}$. A query to $\vw$ can be made through $t$ queries to $\vv^\epsilon$. Consider $W$, the normalization factor such that $\vp_i = \vw_i/W$. If $W$ is known, a query to $\vp$ can be made through a query to $\vw$, which can be done with $t$ queries to $\vv^{\epsilon}$. Since each query to $\vv^\epsilon$ is performed using a row-query to $(\AA,\vb)$, if $W$ is known, a query to $\vp$ at iteration $t$ would require $t$ row-queries to $(\AA,\vb)$.

An important issue is we do not have easy access to $W$. However, if we are able to maintain a constant factor estimation of $W$ throughout the iterations, we can create query access to a vector $\vq$ such that $\vp \leq \vq \leq 2 \vp$. Indeed, assume we have $\widetilde{W} \in [\frac{1}{2}W, W]$, then a query to $\vq = \vw/\widetilde{W}$ cost the same as a query to $\vp$ when $W$ is known, moreover $\vp \leq \vq \leq 2 \vp$.

In order to maintain the estimation of $W$, notice that between two consecutive iterations, the sum of weights stays within a constant factor from its previous value. Hence, we use the constant approximation factor of the previous sum of weights to estimate its new value. This can be done by sampling a constant quantity of constraints using \cref{lem:apers} and Chernoff bound.

\begin{claim} \label{claim:estimation_quantum}
    For an integer $n > 0$, a non-negative vector $\vw \in \rr^n$, a real $W > 0$, a constant $C > 1$, and a probability $p \in [0,1]$. Assume $\widetilde{W} \leq \psnorm[1]{\vw} \leq C \widetilde{W}$, then \cref{ora:q_estimation_mu_better} outputs $\widetilde{W}'$ such that with probability over $1-p$, $\widetilde{W}' \leq \psnorm[1]{\vw} \leq 2\widetilde{W}'$ using
    \begin{equation*}
        \OOt(\sqrt{n} \log \frac{1}{p})
    \end{equation*}
    queries to $\vw$.
\end{claim}

Now that we are able to maintain a constant factor approximation $\widetilde{W}$ of the sum of weights throughout iterations, we want to ensure that the set we are sampling does not have to many constraints. Sampling from $\vq = s \vw/\widetilde{W}$ ensures on average less than $2s$ constraints are sampled, and the set obtained through \cref{lem:apers} satisfies the probability requirement of \cref{lem:lvo-assadi,lem:lvo-clarkson}. However, the set $S$ may have a lot more constraints. We use the following claim to ensure this is not possible with high-probability.   

\begin{claim}
    \label{claim:sampling_quantum}
    For an integer $n > 0$, a non-negative vector $\vw \in \rr^n$, a real $\widetilde{W} > 0$, an integer $s \geq 6$, and a probability $p \in [0,1]$. Assume $\widetilde{W} \leq \psnorm[1]{\vw} \leq 2 \widetilde{W}$. Using $\OOt(\sqrt{n s} \log \frac{1}{p})$ queries to $\vw$, Procedure~\ref{ora:q_sampling} outputs a set $S$ such that with probability over $1-p$:
    \begin{enumerate}
        \item $S$ contains every $i$ independently with probability greater than $\min \lbrace s \frac{\vw_i}{\psnorm[1]{\vw}},1 \rbrace$;
        \item $S$ is small, $\abs{S} \leq 4 s$.
    \end{enumerate}
\end{claim}

First, notice that since we need an algorithm that works with high probability, and we use \cref{claim:estimation_quantum,claim:sampling_quantum} once per iteration, we can assume $p$ is sufficiently small so that the algorithm does not fail because of those procedures with high probability.  Moreover, \cref{claim:estimation_quantum} uses far fewer queries to $\vw$ (hence to $(\AA,\vb)$) compared \cref{claim:sampling_quantum}, hence maintaining a constant factor estimation of normalization factor through the iterations is essentially ``free''.

\subsection{Quantum Speedups for Exact Solvers}

We need another procedure for exact solvers. We need to verify the feasibility of the solution we add to the current set of solutions. If  it is feasible, we return it, else the algorithm should continue. We use Grover search on $\vv$ where $\vv_i(\vx)$ is $1$ if $\vx$ violates constraint $i$, $0$ otherwise. We can use the following lemma from~\cite{wolfQuantumComputingLecture2023}.

\begin{claim}[Exercise 6.c. p60 from~\cite{wolfQuantumComputingLecture2023}]
    \label{claim:quantum_satisfiability}
    Given query access to $\vx \in \lbrace 0, 1 \rbrace^n$, it is possible with $\OO(\sqrt{n \log \frac{1}{p}})$ queries to $\vx$ to find an index $i \in \lbrace 0, \ldots, n-1\rbrace$ such that $\vx_i = 1$ or output $\vx = \zeros$ with probability over $1-p$.
\end{claim}

Since we are performing a linear number of iterations, and we are interested in time complexity up to poly-logarithmic factors, we assume $p$ is sufficiently small so that \cref{claim:quantum_satisfiability} does not fail throughout iterations with high probability.

\begin{algorithm}
    \caption{Quantum Clarkson}
    \label{alg:clarkson-quantum}
    \begin{algorithmic}[1]
        \Require Row-query access to $\AA$ and $\vb$, and $\mathcal{A}$ an exact solver
        \Ensure $\vx$ a feasible solution
        \State $X_0 \gets \emptyset$
        \State $p \gets \frac{1}{32 n \log n} \frac{1}{100 n^2}$
        \State $\widetilde{W}_0 \gets n$
        \State $s \gets 6d^2$
        \For{$t=0$ to $T = 24 d \log n$}
            \State Define the query access to $\vw(X)$ as $2^{\sum_{\vx \in X} \vv^{0}(\vx)}$
\State $S_t \gets \texttt{quantum\_sampling}(\vw(X_t),s,\widetilde{W}_t,p)$ \Comment{Using \cref{ora:q_sampling}.}
\State $\vx \gets \mathcal{A}(S_t)$
            \State $X_{t+1} \gets X_t \cup \lbrace \vx \rbrace$
            \State $\widetilde{W}_{t+1} \gets \texttt{estimation\_sum}(\vw(X_{t+1}),\widetilde{W}_{t},p)$ \Comment{Using \cref{ora:q_estimation_mu_better}.}
        \EndFor

        \For{$\vx \in X_T$}
            \If{$\vx$ is feasible} \Comment{By invoking \cref{claim:quantum_satisfiability}.}
                \State \textbf{return} $\vx$
            \EndIf
        \EndFor
\State \textbf{return} $\bot$
    \end{algorithmic}
\end{algorithm}

\begin{theorem}
    \label{thm:clarkson-quantum}
    For a linear program of the form \ref{eq:general-lp} with $n$ constraints and $d$ variables, \cref{alg:clarkson-quantum} outputs the optimum with high probability using on average $\OOt(\sqrt{n}d^3)$ row-queries to $(\AA,\vb)$, and time
    \begin{equation*}
        \OOt\left(d(\sqrt{n}d^2r + \Texact(d, 24 d^2))\right),
    \end{equation*}
    where $r$ is the row-sparsity of $\AA$, and $\Texact(d,n)$ is the time required to find an exact solution on a set of $n$ constraints with $d$ variables.
\end{theorem}

\subsection{Quantum Speedups for Mixed Packing and Covering LPs}

For mixed packing and covering problems, we can perform the same trick to sample constraints using row-queries to $\CC$. We provide a quantum version of \cref{thm:mpc-classical}.

\begin{theorem}
    \label{thm:mpc-quantum}
    Given a mixed packing and covering problem \ref{prob:mixed-lp}, it is possible to find a $1+\OO(\epsilon)$-approximation $\vx$ or certify infeasibility with high probability using $\OOt(\sqrt{n_c \frac{d}{\epsilon}} \frac{1}{\epsilon^2})$ row-queries to $\CC$, moreover the algorithm runs in time 
    \begin{equation*}
        \OOt \left(\frac{1}{\epsilon} \left(\sqrt{n_c \frac{d}{\epsilon}} \frac{r_c}{\epsilon} + \mathcal{T}_{mixed}(d,\nnz(\PP), \OO\left(\frac{d r_c}{\epsilon} \log (\frac{\log \frac{r_p}{\epsilon}}{\epsilon})\right), \epsilon)\right)\right),
    \end{equation*}
    where $\mathcal{T}_{mixed}(d, \nnz(\PP), \nnz(\CC), \epsilon)$ is the time required to compute a $1+\epsilon$ approximation of a linear program with at most $\nnz(\PP)$ (resp. $\nnz(\CC)$) non-zeros inside the packing (resp. covering) matrix, and $d$ variables.
\end{theorem}

\begin{algorithm}
    \caption{Quantum Mixed Packing and Covering LP Solver}
    \label{alg:q-mixed-epsilon}
    \begin{algorithmic}
        \Require $\PP \in [0,1]^{n_p \times d}$, and quantum query access to $\CC \in [0,1]^{n_c \times d}$, $\epsilon >0$ and $\mathcal{A}$ a mixed packing and covering solver
        \Ensure $\vx$ a $1+\OO(\epsilon)$-approximation
        \State $X_0 \gets \emptyset$
        \State $p \gets \frac{1}{32 n \log n} \frac{1}{100 n^2}$
        \State $\widetilde{W}_0 \gets n_c$
        \State $s \gets 6 \frac{d}{\epsilon} \log \frac{\log (r_p/\epsilon)}{\epsilon}$
        \For{$t=0$ to $T = \frac{24}{\epsilon} \log n_c$}
            \State Define the query access to $\vw(X)$ as $2^{\sum_{\vx \in X} \vv^{\epsilon}(\vx)}$
            \State $S_t \gets \texttt{quantum\_sampling}(\vw(X_t),s,\widetilde{W}_t,p)$ \Comment{Using \cref{ora:q_sampling}.}
\State $\CC_t \gets $ the matrix $\CC$ but only with rows in $S_t$
            \State $\vx \gets \mathcal{A}(\PP, \CC_t, \epsilon)$
            \State $X_{t+1} \gets X_t \cup \lbrace \vx \rbrace$
            \State $\widetilde{W}_{t+1} \gets \texttt{estimation\_sum}(\vw(X_{t+1}),\widetilde{W}_{t},p)$ \Comment{Using \cref{ora:q_estimation_mu_better}.}
        \EndFor
\State \textbf{return} $\sum_{\vx \in X_T} \vx/T$.
    \end{algorithmic}
\end{algorithm}

\subsection{Quantum Speedups for Low-Precision Solvers}

We provide two quantum versions of \cref{thm:lp-classical}. For low-precision solvers, we can perform the same trick to sample constraints using row-queries to $\AA$.

\begin{theorem}
    \label{thm:lp-quantum} 
    For a linear program of the form \ref{eq:general-lp} with largest violation $\Vmax$, \cref{alg:q-lp-epsilon} finds an $\epsilon$-approximation $\vx$ with high-probability using $\OOt(\sqrt{n d \frac{\Vmax}{\epsilon}} \frac{\Vmax^2}{\epsilon^2})$ row-queries to $\AA$, and running time
    \begin{equation*}
        \OOt\left(\frac{\Vmax}{\epsilon} \left(\sqrt{nd \frac{\Vmax}{\epsilon}} \frac{\Vmax}{\epsilon}r + \Texact(d, 24 d \frac{\Vmax}{\epsilon})\right)\right),
    \end{equation*}
    where $r$ is the row-sparsity of $\AA$, and $\Texact(d,n)$ is the running time of an exact solver on $d$ variables and $n$ constraints.
\end{theorem}

Assume we have a set of solutions $\lbrace \vx_j \rbrace_{j=1}^t$, previously, our weight was an exponential on $\sum_{j \in [t]} \vv^{\epsilon}(\vx_j)$. For low-precision solvers, it may be tempting to use a ``continuous'' version that is simpler to compute such as $\frac{1}{ \Vmax} \sum_{j \in [t]} (\AA \vx_j - \vb)$. A query to the weight vector would thus be implemented with one row-query to $(\AA,\vb)$. Unfortunately, we need to be careful about how much each weight can change from one iteration to the other. Both the MWU framework, and our ability to estimate the sum of weights requires at most a constant multiplicative change in the weights, hence we need to use the two-sided width $\rho := \max_{\vx \in \mathcal{D}} \psnorm[\infty]{\AA \vx- \vb}$ instead of the largest violation $\Vmax$ (i.e. one-sided width) used in \cref{thm:lp-classical,thm:lp-quantum}.

\begin{theorem}
    \label{thm:lp-quantum-width}
    For a linear program of the form \ref{eq:general-lp} with width $\rho$, \cref{alg:q-lp-epsilon-width} finds an $\epsilon$-approximation $\vx$ with high-probability using $\OOt(\sqrt{n d \frac{\rho}{\epsilon}} \frac{\rho}{\epsilon})$ row-queries to $\AA$, and running time
    \begin{equation*}
        \OOt\left(\frac{\rho}{\epsilon} \left(\sqrt{nd \frac{\rho}{\epsilon}}r + \Texact(d, 24 d \frac{\rho}{\epsilon})\right)\right),
    \end{equation*}
    where $r$ is the row-sparsity of $\AA$, and $\Texact(d,n)$ is the running time of an exact solver on $d$ variables and $n$ constraints.
\end{theorem} 

\begin{algorithm}
    \caption{Quantum Low-Precision Solver with One-Sided Width}
    \label{alg:q-lp-epsilon}
    \begin{algorithmic}
        \Require Query access to $\AA$ and $\vb$, and $\mathcal{A}$ an exact solver.
        \Ensure $\vx$ an $\epsilon$-approximation.
        \State $X_0 \gets \emptyset$
        \State $p \gets \frac{1}{32 n \log n} \frac{1}{100 n^2}$
        \State $\widetilde{W}_0 \gets n$
        \State $s \gets 6d \frac{\Vmax}{\epsilon}$
        \For{$t=0$ to $T = 24 \frac{\Vmax}{\epsilon} \log n$}
            \State Define the query access to $\vw(X)$ as $2^{\sum_{\vx \in X} \vv^\epsilon(\vx)}$ 
            \State $S_t \gets \texttt{quantum\_sampling}(\vw(X_t),s,\widetilde{W}_t,p)$ \Comment{Using \cref{ora:q_sampling}.}
\State $\vx \gets \mathcal{A}(S_t)$
            \State $X_{t+1} \gets X_t \cup \lbrace \vx \rbrace$
            \State $\widetilde{W}_{t+1} \gets \texttt{estimation\_sum}(\vw(X_{t+1}),\widetilde{W}_{t},p)$ \Comment{Using \cref{ora:q_estimation_mu_better}.}
        \EndFor
\State \textbf{return} $\sum_{\vx \in X_T} \vx/\abs{X_T}$
    \end{algorithmic}
\end{algorithm}

\begin{algorithm}
    \caption{Quantum Low-Precision Solver with Two-Sided Width}
    \label{alg:q-lp-epsilon-width}
    \begin{algorithmic}
        \Require Query access to $\AA$ and $\vb$, and $\mathcal{A}$ an exact solver.
        \Ensure $\vx$ an $\epsilon$-approximation.
        \State $X_0 \gets \emptyset$
        \State $p \gets \frac{1}{32 n \log n} \frac{1}{100 n^2}$
        \State $\widetilde{W}_0 \gets n$
        \State $s \gets 6d \frac{\rho}{\epsilon}$
        \For{$t=0$ to $T = 16 \frac{\rho}{\epsilon} \log n$}
            \State Define the query access to $\vw(X)$ as $2^{\frac{1}{\rho} \left(\AA \sum_{\vx \in X} \vx - t\vb \right)}$ 
            \State $S_t \gets \texttt{quantum\_sampling}(\vw(X_t),s,\widetilde{W}_t,p)$ \Comment{Using \cref{ora:q_sampling}.}
\State $\vx \gets \mathcal{A}(S_t)$
            \State $X_{t+1} \gets X_t \cup \lbrace \vx \rbrace$
            \State $\widetilde{W}_{t+1} \gets \texttt{estimation\_sum}(\vw(X_{t+1}),\widetilde{W}_{t}/2,p)$ \Comment{Using \cref{ora:q_estimation_mu_better}.}
        \EndFor
\State \textbf{return} $\sum_{\vx \in X_T} \vx/\abs{X_T}$
    \end{algorithmic}
\end{algorithm} 

\newpage

\printbibliography
\appendix

\section{Basics of multiplicative weights}\label{sec:mwubasics}
In this section we briefly expand on the argument we developed in \cref{ssec:overview} for why $T = \OO(\rho \log n / \epsilon)$ iterations suffice to generate a sequence of low-violation vectors $\{\vv_t\}_{t=1}^T$ where $\max \left\{\frac{1}{T}\sum_{t=1}^T \vv_T\right\} \leq \frac{\epsilon}{\rho}$.

To do so we prove the following lemma.

\begin{lemma}
Let $f : \Delta_n \rightarrow \{0,1\}^n$ be a (randomized) routine which given any probability distribution $\vp \in \Delta_n$, it outputs a vector $\vv \in \{0,1\}^n$ such that $\langle \vp, \vv\rangle \leq \epsilon$. Then one can provide a sequence $\{\vp_t\}_{t=1}^T$ such that the corresponding outputs $\{\vv_t\}_{t=1}^T$ satisfy
\[
\max \left\{ \frac{1}{T} \sum_{t=1}^T \vv_t \right\} \leq 3\epsilon\,,
\]
where $T = \OO(\log n /\epsilon)$.
\end{lemma}
\begin{proof}
This follows from the standard MWU framework, together with a few useful observations. For the reader's convenience, we provide a complete proof below.
Let $\smax : \mathbb{R}^n \rightarrow \mathbb{R}$, defined as
\[
\smax(\vx) = \log \sum_{i=1}^n \exp (\vx_i)\,.
\]
One can easily verify that $\max(\vx) \leq \smax(\vx) \leq \smax(\vx) + \log n$. Furthermore, we can verify that
\[
\nabla \smax(\vx) = \frac{\exp(\vx)}{ \sum_{i=1}^n \exp (\vx_i)}\,,
\]
and for all $\vx,\vdelta \geq 0$, $\|\vdelta\|_\infty \leq 1$,
\begin{equation}\label{eq:smax_growth}
    \smax(\vx+\vdelta) \leq \smax(\vx) + 2 \cdot \langle \nabla \smax(\vx), \vdelta\rangle \,.
\end{equation}
Using these facts we can define 
\[
\vp_t = \nabla \smax\left( \sum_{i=1}^{t-1}\vv_i \right)\,.
\]
Using the bound from \cref{eq:smax_growth} together with the fact that for all $t$, $\langle \vp_t, \vv_t \rangle \leq \epsilon$ by definition, we have that
\[
\smax\left(\sum_{t=1}^T \vv_t\right) \leq \smax(0) + \sum_{t=1}^T 2 \langle \vp_t, \vv_t\rangle \leq \log n + T \cdot 2\varepsilon\,.
\]
Therefore, 
\[
\max \left\{\frac{1}{T} \sum_{t=1}^T \vv_t\right\} \leq 
\frac{1}{T} \smax\left(  \sum_{t=1}^T \vv_t \right) \leq \frac{1}{T} \left(\log n + T\cdot 2\varepsilon\right)\,.
\]
Thus setting $T = \log n / \varepsilon$ makes the above expression bounded by $3\varepsilon$, which completes the proof.
\end{proof} \section{Recovering Clarkson's and Assadi's results}
\label{app:recovering_results}

\subsection{Recovering Clarkson's Result}

We first show how one can recover the result from Clarkson~\cite{clarksonVegasAlgorithmsLinear1995} to solve an LP by solving $\OO(d \log n)$ times LPs with $\OO(d^2)$ constraints. 

In order to recover Clarkson's result, we have to prove which type of guarantees we would like to have, i.e. what are the desired parameters for $\epsilon$ and $\mu$ in \cref{prob:intro1}. Assume we would like to find an exact solution of an asymmetric linear program with $n$ constraints and $d$ variables. We know there is a set of $d$ constraints such that solving the LP on those $d$ constraints is sufficient. We call this set of constraints $B$.

To solve a linear program exactly, it is sufficient to solve \cref{prob:intro1} with $\mu = 1/d$ and $\epsilon = 0$. Indeed, assume for any $i \in [n]$, one has $\abs{\lbrace j \in [T] : \vx_j \text{ violates constraint } i \rbrace} < \frac{T}{d}$. If we take in particular the $d$ constraints from $B$, we obtain
\begin{equation}
    \label{eq:arg_clarkson}
    \sum_{i \in B} \abs{\lbrace j \in [T] : \vx_j \text{ violates constraint } i \rbrace} < d \cdot \frac{T}{d} = T.
\end{equation}
Assume none of the $\lbrace \vx_j \rbrace_{j=1}^{T}$ is a solution of the LP. That means, each $\vx_j$ violates at least one constraint from $B$. Hence, we have $\sum_{i \in B} \abs{\lbrace j \in [T] : \vx_j \text{ violates constraint } i \rbrace} \geq T$. This contradicts \cref{eq:arg_clarkson} which means at least one vector from $\lbrace \vx_j \rbrace_{j=1}^{T}$ satisfies all the constraint from $B$. Hence, we have an exact solution of the LP. We are now ready to state Clarkson's main result as a corollary of our framework.

\begin{corollary}[Theorem in~\cite{clarksonVegasAlgorithmsLinear1995}]
    \label{cor:clarkson}
    Consider a linear program of the form $\max \langle \vc, \vx \rangle$ subject to $\AA \vx \leq \vb$, where $\AA \in \rr^{n \times d}$. Assume we are given an exact solver $\mathcal{A}$ that runs in time $\Texact(d,n)$ for any linear program on $d$ variables and $n$ constraints, then it is possible to solve exactly the LP in expected time
    \begin{equation*}
        \OO(d \log n \left(\nnz(\AA) + \Texact(d,6d^2)\right)).
    \end{equation*}
\end{corollary}

\subsection{Recovering Assadi's Result}

The other important ingredient is to show that for both vertex cover and minimum odd-set cover, the number of solutions is upper bounded by a small exponential.

\begin{lemma}[Lemma 3.3 and 4.3 from~\cite{assadiSimple$1varepsilon$ApproximationSemiStreaming2025}]
    Consider a graph $G$ with $n$ vertices and $m$ edges.\footnote{For the proof, see the proof of equations (4) and (7) of~\cite{assadiSimple$1varepsilon$ApproximationSemiStreaming2025}.}
    \begin{enumerate}
        \item If $G$ is bipartite, then there are at most $2^n$ distinct candidates for vertex cover.
        \item If $G$ has integral weights upper bounded by $W$, then there are at most $\exp \left(3 n \cdot \log (nW)\right)$ candidates for odd-set cover.
    \end{enumerate}
\end{lemma}

The proof of the first point in the above lemma is due to the fact that in bipartite graphs, minimum vertex cover is integral. For general graphs, there is a lot of structure on the solution of minimum odd-set cover, and the upper bound is obtained through a laminarity property on the set of optimal solutions.

\begin{corollary}[Theorem in \cite{assadiSimple$1varepsilon$ApproximationSemiStreaming2025}]
    Consider a graph $G$ with $n$ vertices and $m$ edges.
    \begin{enumerate}
        \item With exponentially high probability, it is possible to find a $(1+\epsilon)$-approximation of minimum vertex cover on bipartite graphs (thus a $(1-\epsilon)$-approximation of maximum bipartite matching) in time
        \begin{equation*}
            \OO \left(\frac{\log m}{\epsilon} (m + \mathcal{T}(n,\OO(\frac{n}{\epsilon})))\right),
        \end{equation*}
        where $\mathcal{T}(n,m)$ is the time required to compute a vertex cover and a matching on a bipartite graph with $m$ edges and $n$ vertices.
        \item If $G$ has integral weights in $[0, W]$, with exponentially high probability it is possible to find a $(1+\epsilon)$-approximation of minimum odd-set cover on graphs (thus a $(1-\epsilon)$-approximation of maximum general matching) in time
        \begin{equation*}
            \OO \left(\frac{\log W}{\epsilon} (m + \mathcal{T}(n,\OO(\frac{n \log (nW)}{\epsilon})))\right),
        \end{equation*}
        where $\mathcal{T}(n,m)$ is the time required to compute a minimum odd-set cover and a maximum weight matching on a graph with $m$ edges and $n$ vertices.
    \end{enumerate}
\end{corollary} 
\section{Proofs of \texorpdfstring{\cref{sec:applications}}{section Applications}}
\label{apps:applications}
\subsection{Proof of \texorpdfstring{\cref{lem:lvo-clarkson}}{lemma }}

We first prove the following lemma which states that the number of constraint violated by an exact solution is inversely proportional to the number of constraint sampled.

\begin{lemma}[Sampling lemma \cite{clarksonVegasAlgorithmsLinear1995}]
    \label{lem:sampling2}
    Consider a linear program of the form $\max \langle \vc, \vx \rangle$ subject to $\AA \vx \leq \vb$, where $\AA \in \rr^{n \times d}$. Given a probability distribution $\vp \in \Delta_n$, consider a set $S \subseteq [n]$ such that
    \begin{equation*}
        \mathbb{P}\left[i \in S\right] \geq \min \lbrace r \vp_i, 1 \rbrace.
    \end{equation*}
    Let $\vx$ be a solution of the partial LP defined by the constraint from $S$. We have 
    \begin{equation*}
        \mathbb{E}\left[\sum_{i \in [n]: \langle \AA_{i:}, \vx \rangle > \vb_i} \vp_i \right] \leq \frac{d}{r}.
    \end{equation*}
\end{lemma}

\begin{proof}
    For any subset of constraint $S \subseteq [n]$, we define $\vx^*_S$ as the exact solution of the LP defined on those constraints. For simplicity, we define the violation vector $\vv(\vx)$ as:
    \begin{equation*}
        \vv_i(\vx) :=
        \begin{cases}
            1 & \text{if } \langle \AA_{i:}, \vx \rangle > \vb_i \qquad \text{(constraint } i \text{ is violated);}\\
            0 & \text{if }  \langle \AA_{i:}, \vx \rangle \leq \vb_i \qquad \text{(constraint } i \text{ is satisfied).}
        \end{cases}
    \end{equation*}
    We have 
    \begin{align*}
        \mathbb{E}\left[\langle \vp, \vv(\vx) \rangle\right] &= \sum_{S \subseteq [n]} \mathbb{P}\left[S\right] \mathbb{E}\left[\langle \vp, \vv(\vx) \rangle \vert S\right] \\
        &= \sum_{S \subseteq [n]} \mathbb{P}\left[S\right] \mathbb{E}\left[\langle \vp, \vv(\vx^*_S) \rangle \right] \\
        &= \sum_{S \subseteq [n]} \mathbb{P}\left[S\right] \sum_{i \in [n] \setminus S} \vp_i \vv_i(\vx^*_S) \\
        &= \sum_{S \subseteq [n]} \sum_{i \in [n] \setminus S} \mathbb{P}\left[S\right]  \vp_i \vv_i(\vx^*_S) \\
        &\leq \frac{1}{r} \sum_{S \subseteq [n]} \sum_{i \in [n] \setminus S} \mathbb{P}\left[S \cup \lbrace i \rbrace \right] \vv_i(\vx^*_S) \\
        &= \frac{1}{r} \sum_{Q \subseteq [n]} \mathbb{P}\left[Q\right] \sum_{i \in Q} \vv_i(\vx^*_{Q \setminus \lbrace i \rbrace}).
    \end{align*}
    Since $\vx^*_Q$ is fully determined by $d$ constraints from $Q$, $\vv_i(\vx^*_{Q \setminus \lbrace i \rbrace})$ equals $1$ on those constraints only. Hence,
    \begin{equation*}
        \mathbb{E}\left[\langle \vp, \vv(\vx) \rangle\right] \leq \frac{1}{r} \sum_{Q \subseteq [n]} \mathbb{P}\left[Q\right] d = \frac{d}{r}.
    \end{equation*}
\end{proof}

We are now ready to prove \cref{lem:lvo-clarkson}.

\begin{proof}[Proof of \cref{lem:lvo-clarkson}]
    We use \cref{lem:sampling2} with $r = 2d/\mu$, hence we have that on average $\langle \vp, \ones_{\AA \vx > \vb} \rangle \leq \frac{\mu}{2}$. Using Markov's inequality, with probability $1/2$, we have $\langle \vp, \ones_{\AA \vx > \vb} \rangle \leq \mu$. Hence, on average $2$ iterations are necessary, which means the expected time is $\OO(\nnz(\AA) + \Texact(d,2\frac{d}{\mu}))$.
\end{proof}

\subsection{Proof of \texorpdfstring{\Cref{lem:lvo-assadi}}{lemma }}

\begin{proof}
   For any $\vx \in \bigcup_{S \subseteq [n]} \mathcal{A} \langle S \rangle$, consider $Q(\vx) := \lbrace i \in [n] : \langle \AA_{i:}, \vx \rangle > \vb_i + \epsilon \rbrace$. Assume $\sum_{i \in Q(\vx)} \vp_i > \mu$. We are sampling each constraint with probability at least $r \vp_i$ where $r = \frac{\log (Nn)}{\mu}$. We have
    \begin{align*}
        \mathbb{P}[\vx \text{ is } \epsilon\text{-feasible on the sampled constraints}] &= \mathbb{P}[\text{No constraints from } Q(\vx) \text{ were sampled}] \\
        &\leq \prod_{i \in Q(\vx)} (1-r \vp_i) \\
        &\leq \exp \left(- \sum_{i \in Q(\vx)} r \vp_i\right) \\
        &\leq \exp \left( -r \mu \right) \leq \frac{1}{Nn}.
    \end{align*}

    A union bound over the $N$ possible outputs of the solver $\mathcal{A}$ ensures that with high-probability the sum of probabilities of violated constraints is smaller than $\mu$.
\end{proof}

\subsection{Proof of \texorpdfstring{\Cref{cl:sparsify-mixed}}{Claim }}

\begin{proof}
    Assume we have $\vx$ such that $\CC \vx \geq \ones$, and $\PP \vx \leq (1+ \epsilon) \ones$. We will discretize each entry of $\vx$. We compute our new solution $\tilde{\vx}$ as follows:
    \begin{equation*}
        \tilde{\vx}_i := \begin{cases}
            0 &\text{if } \vx_i = 0;\\
            \frac{\epsilon}{r_p} &\text{if } 0 < \vx_i \leq \frac{\epsilon}{r_p};\\
            \min \lbrace \frac{\epsilon}{r_p} (1+\epsilon)^{k+1},1 \rbrace &\text{if } \frac{\epsilon}{r_p} (1+\epsilon)^k < \vx_i \leq \min \lbrace \frac{\epsilon}{r_p} (1+\epsilon)^{k+1},1 \rbrace.
        \end{cases}
    \end{equation*}

    $\tilde{\vx}$ verifies $\zeros \leq \vx \leq \tilde{\vx} \leq \ones$ where the inequalities are taken coordinate-wise.
    
    \paragraph{$\tilde{\vx}$ is a $1 + 4\epsilon$ approximation.} Since $\vx \leq \tilde{\vx}$ coordinate-wise, we have $\CC \tilde{\vx} \geq \CC \vx \geq \ones$.

    We need to verify that $\PP \tilde{\vx} \leq (1+4 \epsilon) \ones$. We have for any $i$, $\tilde{\vx}_i \leq (1+\epsilon) \vx_i + \frac{\epsilon}{r_p}$. Hence, 
    \begin{align*}
        \PP \tilde{\vx} \leq (1+\epsilon) \PP \vx + \frac{\epsilon}{r_p} \PP \ones \leq (1+ \epsilon)^2 \ones + \epsilon \ones \leq (1+4\epsilon) \ones.
    \end{align*}
    Hence, $\tilde{\vx}$ is a $1+4\epsilon$ approximation.

    \paragraph{Size of $\mathcal{S}$.} For each coordinate $i$, we have $\vx_i \in \lbrace 0,1 \rbrace \cup \lbrace \frac{\epsilon}{r_p} (1+\epsilon)^k : k \in [K] \rbrace$. Where $K$ is the largest integer such that $\frac{\epsilon}{r_p}(1+\epsilon)^K < 1$. Hence, $K \leq \lceil 4 \frac{\log (r_p/\epsilon)}{\epsilon} \rceil$. Since there are $d$ coordinates, we have $\abs{\mathcal{S}} \leq {\lceil 2 + 4\frac{\log (r_p/\epsilon)}{\epsilon} \rceil}^d$.  

\end{proof}

\subsection{Proof of \texorpdfstring{\cref{thm:mpc-classical}}{Theorem }}

\begin{proof}[Proof of \cref{thm:mpc-classical}]
    We prove this theorem in three steps. First, we prove the requirement on the $(\epsilon,\mu)$-low-violation oracle. Second, we prove that $\OO(\frac{\log n_c}{\epsilon})$ iterations are sufficient. Finally, we construct our $(\epsilon,\mu)$-low-violation oracle, and show the expected running time of each iteration.

    \paragraph{The required $(\epsilon,\mu)$-low-violation oracle.} Assume we have $\lbrace \vx_j \rbrace_{j \in [T]}$ a solution of \cref{prob:intro1}. Then, since at each iteration we are using a mixed packing and covering solver that contains all the packing constraints, we have for all $j \in [T]$, $\PP \vx_j \leq (1+\epsilon)\ones$. Hence, for the packing constraints, $\bar{\vx} := \sum_{j \in [T]} \vx_j / T$ already satisfies $\PP \bar{\vx} \leq (1+\epsilon)\ones$. For the covering constraints, for $j \in [T]$, we have $\CC \vx_j \geq \zeros$. 

    For a covering constraint $i$, let $V_i \subseteq [T]$ be the set of solutions violating constraint $i$: $\lbrace \vx_j : j \in [T] \wedge \langle \CC_{i:}, \vx_j \rangle < 1 \rbrace$. We have
    \begin{equation*}
        \sum_{j \in [T]} \langle \CC_{i:}, \vx_j \rangle -  1 = \sum_{j \in V_i}  \langle \CC_{i:}, \vx_j \rangle - 1 + \sum_{j \in [T] \setminus V_i}  \langle \CC_{i:}, \vx_j \rangle - 1 \leq \sum_{j \in V_i} 0 + \sum_{j \in [T] \setminus V_i} 1 \geq (1 - \mu)T.
    \end{equation*}
    For $\mu = \epsilon$, $\bar{\vx}/(1-\epsilon)$ is a $(1+4\epsilon)$-approximation. 

    \paragraph{Number of iterations.} Since we are using a $(\epsilon,\epsilon)$-low-violation oracle, and we are sampling from a set of at most $n_c$ constraints, $\OO(\frac{\log n_c}{\epsilon})$ iterations are required per \cref{thm:main-theorem}.

    \paragraph{Constructing the $(\epsilon,\epsilon)$-low-violation oracle.} We want to use \cref{lem:lvo-assadi}, we need to create a mixed packing and covering solver that $1+\epsilon$-approximations from a small set of vectors. Given any mixed packing and covering solver $\mathcal{A}$, we post-process a call on $\mathcal{A}$ to discretize the output using \cref{cl:sparsify-mixed}. Formally, given any set of constraints $S$, we construct the $1+4\epsilon$-approximate algorithm: \begin{enumerate}
        \item Compute $\vx$ the output of $\mathcal{A}(S)$.
        \item Compute $\tilde{\vx}$ using \cref{cl:sparsify-mixed}.
    \end{enumerate}
    Our new solver has outputs in a set with at most ${\lceil 2 + 4\frac{\log (r_p/\epsilon)}{\epsilon} \rceil}^d$ elements. Using \cref{lem:lvo-assadi}, it is sufficient to sample on average $\OO\left(\frac{d}{\epsilon} \log (1+ \frac{\log (\frac{r_p}{\epsilon})}{\epsilon})\right)$ constraints.

    Hence, at each iteration, we need to run a mixed packing and covering solver with packing constraint $\PP$, and $\OO\left(\frac{d}{\epsilon} \log (1+ \frac{\log (\frac{r_p}{\epsilon})}{\epsilon})\right)$ covering constraints, each one having row-sparsity at most $r_c$. Moreover, at each iteration, we need to update the sampling probabilities which cost $\OO(\nnz(\CC))$.
\end{proof}

\section{Proofs of \texorpdfstring{\Cref{sec:quantum}}{ Section Quantum}}
\label{apps:quantum_claims}

\subsection{Proofs of \texorpdfstring{\Cref{claim:estimation_quantum,claim:sampling_quantum}}{Claims}}

Notice that using \cref{lem:apers} with $\vq = s \frac{\vw}{\widetilde{W}}$ outputs a set with on average $s \frac{\psnorm[1]{\vw}}{\widetilde{W}}$ elements. Hence, we can estimate $\psnorm[1]{\vw}$ by sampling multiple times using \cref{lem:apers} and taking the set with the median number of elements. We describe in the following claim a procedure to enhance the guarantees of \cref{lem:apers} using multiple call to it.

\begin{claim}
    For a non-negative vector $\vw \in \rr^n$, a real $\widetilde{W} \leq \psnorm[1]{\vw}$, an integer $s \geq 6$, and a probability $p$. Assume one has quantum query access to $\vw$. Then, using $\OOt(\sqrt{n s \psnorm[1]{\vw}/\widetilde{W}} \log \frac{1}{p})$ queries to $\vw$, \cref{ora:q_sampling} computes a set $S$ such that:
    \begin{enumerate}
        \item Every $i \in [n]$ is in $S$ with probability over $\min \lbrace s \frac{\vw_i}{\widetilde{W}}, 1 \rbrace$;
        \item The size of $S$ is close to its average size:
        \begin{equation*}
            \mathbb{P}\left[\abs{\abs{S} - s\frac{\psnorm[1]{\vw}}{\widetilde{W}}} \geq \sqrt{6 s\frac{\psnorm[1]{\vw}}{\widetilde{W}}}\right] \leq p.
        \end{equation*}
    \end{enumerate}

\end{claim}

\begin{procedure}
\caption{$\texttt{quantum\_sampling}(\vw, s, \widetilde{W},p)$}
    \label{ora:q_sampling}
    \begin{algorithmic}[1]
        \Require Quantum query access to a non-negative vector $\vw$, $s \geq 6$, $\widetilde{W} \leq \psnorm[1]{\vw} = \OO(\widetilde{W})$, $p$ a probability
        \Ensure $\bar{S} \subseteq [n]$ such that $\bar{S}$ contains each $i$ independently with probability $\frac{s}{\widetilde{W}}\vw_i$.
        \For{$r = 1$ to $1+ 5 \log \frac{1}{p}$}
            \State $S_r \gets$ \cref{lem:apers} with probabilities $\frac{s}{\widetilde{W}} \vw$
            \State $s_r \gets \abs{S_r}$
        \EndFor
        \State $\bar{S} \gets S_{\bar{r}}$ such that $s_{\bar{r}}$ is the median of $(s_r)_r$ 
        \State \textbf{return} $\bar{S}$
    \end{algorithmic}
\end{procedure}

\begin{proof}
    Consider $X_1, \ldots, X_n$ independent random variables taking values in $\lbrace 0,1 \rbrace$. Let $X$ denote their sum and let $\mu = \mathbb{E}\left[X\right]$ denote the sum's expected value. Then for any $0 \leq \delta \leq 1$, multiplicative Chernoff bound states that:
    \begin{equation}
        \label{eq:multiplicative-Chernoff-bound}
        \mathbb{P}\left[\abs{X - \mu} \geq \delta \mu\right] \leq 2 e^{-\frac{\delta^2 \mu}{3}}.
    \end{equation}

    Assume $S$ is the output of a call to \cref{lem:apers} using $\vq = s \vw/ \widetilde{W}$. Let $X_i$ be the event of constraint $i$ being in $S$. $X_i \in \lbrace 0,1 \rbrace$, and $\abs{S} = \sum_{i} X_i = X$. Moreover, $\mathbb{E}\left[X\right] = s \psnorm[1]{\vw}/\widetilde{W}$. We have using \cref{eq:multiplicative-Chernoff-bound} with $\delta = \sqrt{\frac{6}{\mathbb{E}\left[X\right]}} \leq 1$ since $s \geq 6$:
    \begin{equation}
        \label{eq:lem-apers-bound}
        \mathbb{P}\left[\abs{\abs{S} - s \frac{\psnorm[1]{\vw}}{\widetilde{W}}} \geq \sqrt{6 s\frac{\psnorm[1]{\vw}}{\widetilde{W}}}\right] \leq 2e^{-2} \leq 1/e.
    \end{equation}

    Consider $\bar{S}$ the output of \cref{ora:q_sampling}. In \cref{ora:q_sampling}, we have $\OO(\log \frac{1}{p})$ sets satisfying \cref{eq:lem-apers-bound}. Since $\bar{S}$ has the size of the median of those $R$ sets, if $\bar{S}$ does not satisfy \cref{eq:lem-apers-bound}, then at least $R/2$ sets from $(S_r)_{r \in [R]}$ do not satisfy it. This happens with probability less than $1/e^{R/2}$. Hence, since $R = \lceil 5 \log \frac{1}{p} \rceil$, this happens with probability less than $p$. Hence:
    \begin{equation*}
        \mathbb{P}\left[\abs{\abs{\bar{S}} - s \frac{\psnorm[1]{\vw}}{\widetilde{W}}} \geq \sqrt{6 s \frac{\psnorm[1]{\vw}}{\widetilde{W}}}\right] \leq p.
    \end{equation*}
\end{proof}

\paragraph{Proof of \cref{claim:estimation_quantum}.}

\begin{procedure}
    \caption{$\texttt{estimation\_sum}(\vw,\widetilde{W},p)$}
    \label{ora:q_estimation_mu_better}
    \begin{algorithmic}[1]
        \Require Quantum query access to a non-negative vector $\vw$, $\widetilde{W} \leq \psnorm[1]{\vw} = \OO(\widetilde{W})$, $p$ a probability
        \Ensure $\widetilde{W}'$ such that $\widetilde{W}' \leq \psnorm[1]{\vw} \leq 2 \widetilde{W}'$ with probability at least $1-p$
        \State $s \gets 71$
        \State $S \gets \texttt{quantum\_sampling}(\vw,s,\widetilde{W},p)$ \Comment{Using \cref{ora:q_sampling}}
        \State $\widetilde{W}' \gets \frac{\widetilde{W}}{2 s}\left(\sqrt{3 + 2 \abs{S}} - \sqrt{3}\right)^2$
\State \textbf{return} $\widetilde{W}'$
    \end{algorithmic}
\end{procedure}

\begin{proof}
    We will use \cref{ora:q_estimation_mu_better}, with probability at least $1-p$, we have 
    \begin{equation}
        \label{eq:bound-size-S}
        s \frac{\psnorm[1]{\vw}}{\widetilde{W}} -  \sqrt{6 s \frac{\psnorm[1]{\vw}}{\widetilde{W}}} \leq \abs{S} \leq  s \frac{\psnorm[1]{\vw}}{\widetilde{W}} + \sqrt{6 s \frac{\psnorm[1]{\vw}}{\widetilde{W}}}.
    \end{equation}
    In particular, $\abs{S} \leq s \frac{\psnorm[1]{\vw}}{\widetilde{W}} + \sqrt{6 s \frac{\psnorm[1]{\vw}}{\widetilde{W}}}$. Hence, with $x = \sqrt{\psnorm[1]{\vw}}$, we have 
    \begin{equation*}
        0 \leq \frac{s}{\widetilde{W}} x^2 + \sqrt{6 \frac{s}{\widetilde{W}}}x - \abs{S}.
    \end{equation*}
    Hence, we have $x \geq \sqrt{\frac{\widetilde{W}}{2s}} (\sqrt{3 +2 \abs{S}} - \sqrt{3})$. Let $\widetilde{W}' := \frac{\widetilde{W}}{2s}(\sqrt{3+ 2 \abs{S}} - \sqrt{3})^2$, we have $\widetilde{W}' \leq x^2 = \psnorm[1]{\vw}$.

    On the other hand, $s \frac{\psnorm[1]{\vw}}{\widetilde{W}} -  \sqrt{\frac{6 s \psnorm[1]{\vw}}{\widetilde{W}}} \leq \abs{S}$. Let $x = \sqrt{\psnorm[1]{\vw}}$, the inequality implies $x \leq \sqrt{\frac{\widetilde{W}}{2s}} (\sqrt{3+2 \abs{S}} + \sqrt{3})$. Hence, 
    \begin{equation*}
        \frac{\psnorm[1]{\vw}}{\widetilde{W}'} \leq \left(\frac{\sqrt{3+2 \abs{S}} + \sqrt{3}}{\sqrt{3+2 \abs{S}} - \sqrt{3}}\right)^2 = \left(\frac{6 + 2 \abs{S} + 2 \sqrt{9 + 6 \abs{S}}}{2\abs{S}}\right)^2 = \left(1 + \sqrt{\frac{6}{\abs{S}} + \frac{9}{\abs{S}^2}} + \frac{3}{\abs{S}}\right)^2.
    \end{equation*}
    Using the above inequality, one can verify if $\abs{S} \geq 50$, then $\psnorm[1]{\vw} \leq 2 \widetilde{W}'$.

    It remains to show $\abs{S}$ has at least $50$ elements with probability larger than $1-p$. We assume $\abs{S}$ satisfied \cref{eq:bound-size-S}. Hence, 
    \begin{equation*}
        \abs{S} \geq s \frac{\psnorm[1]{\vw}}{\widetilde{W}} - \sqrt{6 s \frac{\psnorm[1]{\vw}}{\widetilde{W}}} = s \frac{\psnorm[1]{\vw}}{\widetilde{W}} \left(1 - \sqrt{6 \frac{\widetilde{W}}{s \psnorm[1]{\vw}}}\right) \geq s \left(1 - \sqrt{\frac{6}{s}}\right) = s - \sqrt{6s}.
    \end{equation*}
    Where the last inequality comes from $\widetilde{W} \leq \psnorm[1]{\vw}$. Hence, with $s \geq 71$, we have $\abs{S} \geq 50$.
\end{proof}

\paragraph{Proof of \cref{claim:sampling_quantum}.}
\begin{proof}
    We use \cref{ora:q_sampling} with $s \geq 6$, hence with probability at least $1-p$:
    \begin{equation*}
        \abs{S} \leq s \frac{\psnorm[1]{\vw}}{\widetilde{W}} + \sqrt{6 s \frac{\psnorm[1]{\vw}}{\widetilde{W}}}
        = s \frac{\psnorm[1]{\vw}}{\widetilde{W}} \left(1 + \sqrt{\frac{6}{s} \frac{\widetilde{W}}{\psnorm[1]{\vw}}}\right)
        \leq s \frac{\psnorm[1]{\vw}}{\widetilde{W}} \left(1 + \sqrt{\frac{6}{6} \frac{\psnorm[1]{\vw}}{\psnorm[1]{\vw}}}\right)
        = 2 s \frac{\psnorm[1]{\vw}}{\widetilde{W}}
        \leq 4 s.
    \end{equation*}
\end{proof}

\subsection{Proof of \texorpdfstring{\Cref{thm:clarkson-quantum}}{Theorem }}

\begin{proof}
    As per \cref{cor:clarkson}, it is sufficient to be able to sample in $S$ constraint $i$ with probability larger than $\min \lbrace 6d^2 \vp_i, 1 \rbrace$. At iteration $t$, we have a set $X$ such that $\abs{X} = t$. Let $\vw_i(X) := 2^{\sum_{\vx \in X} \vv_i^0(\vx)}$ where $\vv^0(\vx)$ is the binary violation vector, i.e. $\vv^0_i(\vx) = 0$ if $\vx$ satisfies constraint $i$ ($\langle \AA_{i:}, \vx \rangle \leq \vb_i$), and $1$ if $\vx$ violates constraint $i$.

    \paragraph{Maintenance of the estimation of $\psnorm[1]{\vw(X_t)}$:} We first prove at iteration $t$, we have $\widetilde{W}_t \leq \psnorm[1]{\vw(X_t)} \leq 2 \widetilde{W}_t$. 
    
    For $t =0$, we have $\widetilde{W}_0 = n = \psnorm[1]{\vw(X_0)} = \psnorm[1]{\vw(\emptyset)}$.

    Assume for $t \geq 0$ we have $\widetilde{W}_t \leq \psnorm[1]{\vw(X_t)} \leq 2 \widetilde{W}_t$ where $X$ contains $t$ elements. Consider $X_{t+1} = X_t \cup \lbrace \vx \rbrace$, For each $i \in [n]$, $\vw_i(X_{t}) \leq \vw_i(X_{t+1}) \leq 2 \vw_i(X_t)$. Hence, we have $\psnorm[1]{\vw(X_t)} \leq \psnorm[1]{\vw(X_{t+1})} \leq 2 \psnorm[1]{\vw(X_t)}$ which means $\widetilde{W}_t \leq \psnorm[1]{\vw(X_{t+1})} \leq 4 \widetilde{W}_t$. When we use \cref{ora:q_estimation_mu_better} to estimate $\psnorm[1]{\vw(X_{t+1})}$, the conditions of \cref{claim:estimation_quantum} are satisfied, hence $\widetilde{W}_{t+1}$ is such that $\widetilde{W}_{t+1} \leq \psnorm[1]{\vw(X_{t+1})} \leq 2 \widetilde{W}_{t+1}$.

    Hence, we always have $\widetilde{W}_{t} \leq \psnorm[1]{\vw(X_t)} \leq 2 \widetilde{W}_t$.

    \paragraph{Low-violation solution:} Since $\widetilde{W}_t \leq \psnorm[1]{\vw(X_t)}$, we have:
    \begin{equation*}
        \sum_i \vw(X_{t+1}) = \sum_{i} \vw_i(X_t) (1+ \vv_i^0(\vx))
        = \psnorm[1]{\vw(X_t)}\left(1 + \sum_{i} \frac{\vw_i(X_t)}{\psnorm[1]{\vw(X_t)}} \vv_i^{0}(\vx)\right).
    \end{equation*}
    Moreover, since $S_t$ is sampled using \cref{claim:sampling_quantum}, we know that each constraint $i$ is sampled independently and is present in $S_t$ with probability at least $6d^2 \vw_i(X_t)/\psnorm[1]{\vw(X_t)}$. Hence, with \cref{lem:sampling2}, we have $\langle \frac{\vw(X_t)}{\psnorm[1]{\vw(X_t)}}, \vv^0(\vx) \rangle \leq \frac{1}{6d}$.

    Hence, $\mathbb{E}\left[\psnorm[1]{\vw(X_{t+1})}\right] \leq (1+\frac{1}{6d}) \psnorm[1]{\vw(X_t)}$. Or:
    \begin{equation*}
        \mathbb{E}\left[\psnorm[1]{\vw(X_{T})}\right] \leq (1+\frac{1}{6d})^T n \leq 2^{\frac{1}{2} \frac{T}{d}}n.
    \end{equation*}

    \paragraph{Returning the solution:} If there is $\vx \in X_T$ such that $\vx$ is a solution to the linear program, then it should be detected with high probability with the call to \cref{claim:quantum_satisfiability}. Hence, we can assume that none of the $\vx \in X_T$ are solutions to the LP. 

    Since none of the $\vx$ in $X_T$ is a solution to the linear program, at least one constraint from the basis $B$ (i.e. the set of $d$ constraints that described the optimum) is violated at each iteration. Hence, $\sum_{i \in B} \vw_i(X_T) \geq 2^{\frac{T}{d}}$.

    For some constant $c> 0$, using Markov's inequality, we have with probability at least $1- n^{-c}$, $\psnorm[1]{\vw(X_T)} \leq n^{c+1} 2^{\frac{3}{2} \frac{T}{d}}$. Hence, with high probability we can assume $\psnorm[1]{\vw(X_T)} \leq n^{3} 2^{\frac{1}{2} \frac{T}{d}}$. Thus, as long as none of the $\vx$ in $X_T$ are solutions to the LP, the following should be true with high probability:
    \begin{equation*}
        2^{\frac{T}{d}} \leq \psnorm[1]{\vw(X_T)} \leq n^{3} 2^{\frac{1}{2} \frac{T}{d}}.
    \end{equation*}
    Hence, $T \leq 6d \log_2 n \leq 24 d \log n$. This is a contradiction since we are performing more than $24 d \log n$ iterations.
    
    \paragraph{Number of row-queries to $(\AA,\vb)$:} Within iteration $t$, since $X_t$ has less than $T$ elements, hence a query to $\vw(X_t)$ can be done using at most $T$ queries to $\vv^0$, which correspond to $T$ row-queries to $(\AA,\vb)$.

    We will count the number of queries to $\vw(X_t)$. During iteration $t$, \cref{claim:estimation_quantum,claim:quantum_satisfiability} perform at most $\OOt(\sqrt{n})$ queries to $\vw(X_t)$. On the other hand, \cref{claim:sampling_quantum} performs $\OOt(\sqrt{n d^2})$ queries to $\vw(X_t)$. 

    Since each query to $\vw(X_t)$ requires at most $T$ row-queries to $(\AA,\vb)$, we indeed have that in total $\OOt(\sqrt{n}d^3)$ row-queries to $(\AA,\vb)$ are performed.

    \paragraph{Running time:} $\OOt(d)$ iterations are performed, each of which requires at most $\OOt(\sqrt{n d^2}d)$ row-queries to $(\AA,\vb)$ to sample the set of constraints. This requires time $\OOt(\sqrt{n}d^2 r)$ where $t$ is the row-sparsity of $\AA$. Moreover, one call to an exact solver is used. Since $\abs{S_t} \leq 4 s = 24 d^2$, we have that the call is made on an exact solver on at most $24 d^2$ constraints and $d$ variables.
\end{proof}

\subsection{Proof of \texorpdfstring{\cref{thm:mpc-quantum}}{Theorem }}

\begin{proof}
    The proof is similar to \cref{thm:mpc-classical}, the difference is on how we sample the constraints at each iteration. We use the exact same techniques as in \cref{thm:clarkson-quantum}, and obtain a query access to $\vw(X)/\psnorm[1]{\vw(X)}$ at each iteration by maintaining a constant factor estimation of $\psnorm[1]{\vw(X)}$ and using queries to rows of $\AA$, and $\vb$.
\end{proof}

\subsection{Proof of \texorpdfstring{\cref{thm:lp-quantum-width}}{Theorem }}

\begin{proof}
    In this theorem, the proof is different since the value of $\vw$ is not the same. For the sake of the proof, we do not use MWU in a black-box fashion, since this problem does not fit inside \cref{prob:intro1}.

    At iteration $t$, we have $\vw_i(X_t) = 2^{\sum_{\vx \in X_t} \frac{\langle \AA_{i:}, \vx \rangle - \vb_i}{\rho}}$. Notice that one query to $\vw(X)$ can be made using one row-query to $(\AA,\vb)$, indeed, we can maintain $\bar{\vx} := \sum_{\vx \in X_t} \vx$ and use $\vw(X_t) = 2^{\frac{1}{\rho} \left(\AA \bar{\vx} - t \vb\right)}$.
    
    We consider $\vp = \vw(X)/\psnorm[1]{\vw(X)}$. If we seek to sample $s$ elements from $[n]$ with the probability distribution $\vp$, we require $\OOt(\sqrt{n s})$ queries to $\vp$, each one requires one row-query to $(\AA,\vb)$. 

    At each iteration, we use an average $(\epsilon,\epsilon/3\rho)$-low-violation oracle with $\vp$ to obtain $\tilde{\vx}$. We have $\langle \vp, \vv^{\epsilon}(\tilde{\vx}) \rangle \leq \mu$, hence on average:
    \begin{align*}
        \psnorm[1]{\vw(X \cup \lbrace \tilde{\vx} \rbrace)} &= \sum_{i \in [n]} 2^{\sum_{\vx \in X} \frac{\langle \AA_{i:}, \vx \rangle - \vb_i}{\rho}} 2^{\frac{\langle \AA_{i:}, \tilde{\vx} \rangle - \vb_i}{\rho}} \\ 
        &\leq \sum_{i \in [n]} 2^{\sum_{\vx \in X} \frac{\langle \AA_{i:}, \vx \rangle - \vb_i}{\rho}} \left(1 + \vv_i^\epsilon(\tilde{\vx})\right) \\
        &\leq \left(1 + \frac{\epsilon}{3\rho}\right) \sum_{i \in [n]} 2^{\sum_{\vx \in X} \frac{\langle \AA_{i:}, \vx \rangle - \vb_i}{\rho}} \\
        &= \left(1 + \frac{\epsilon}{3\rho}\right)  \psnorm[1]{\vw(X)} \leq 2^{\frac{1}{2}\frac{\epsilon}{\rho}} \psnorm[1]{\vw(X)}.
    \end{align*}
    Therefore, we have $\mathbb{E}\left[\psnorm[1]{\vw(X_T)}\right] \leq 2^{\frac{T}{2} \frac{\epsilon}{\rho}} n$. Hence, using Markov's inequality, with high-probability we have $\psnorm[1]{\vw(X_T)} \leq 2^{\frac{T}{2} \frac{\epsilon}{\rho}} n^2$.

    For any constraint $i$, we can use $\vw_i(X_T)$ to compute its violation, indeed, $\left(\AA \frac{\bar{\vx}}{T} - \vb\right)_i = \frac{\rho}{T} \log_2(\vw_i(X_T))$. Hence, it is sufficient to show that for any $i$, $\log_2 \vw_i(X_T) \leq T \frac{\epsilon}{\rho}$.

    We have:
    \begin{equation*}
        \log_2 \vw_i(X_T) \leq \log_2 \left(\sum_{j} \vw_j(X_T)\right) = \log_2 \psnorm[1]{\vw(X_T)} \leq \frac{T}{2} \frac{\epsilon}{\rho} + 2 \log_2 n \leq T\frac{\epsilon}{\rho},
    \end{equation*}
    where the last inequality comes from $2 \log_2 n \leq \frac{T}{2} \frac{\epsilon}{\rho}$ since $T \geq 16 \frac{\rho}{\epsilon} \log n$.

    It remains to prove we can estimate $\widetilde{W}_{t+1}$ using $\texttt{estimation\_sum}(\vw(X_{t+1}),\widetilde{W}_{t}/2,p)$. Indeed, in our case the weights can decrease, however, since we are using $\rho$ instead of $\Vmax$, the decrease can not be too large. In fact, we have $\frac{1}{2} \vw_i(X_t) \leq \vw_i(X_{t+1}) \leq 2 \vw_i(X_t)$. Therefore, we have $\widetilde{W}_{t+1}$ such that $\widetilde{W}_{t+1} \leq \sum_i \vw_i(X_{t+1}) \leq 2 \widetilde{W}_{t+1}$.
\end{proof}

\end{document}